\documentclass[onefignum,onetabnum]{siamonline181217}

\usepackage{enumitem}
\usepackage{threeparttable}
\usepackage{siunitx}
\usepackage{lipsum}
\usepackage{amsfonts}
\usepackage{graphicx}
\usepackage{subcaption}
\usepackage{graphbox}
\usepackage{epstopdf}
\usepackage{algorithmic}
\usepackage{booktabs}
\usepackage{dsfont}
\usepackage[colorinlistoftodos]{todonotes}
\usepackage{xcolor}
\usepackage{bbm}
\usepackage{mathtools}
\usepackage{amsmath,bm}
\usepackage{verbatim}
\usepackage{wrapfig}
\usepackage{tikz}
\usepackage{soul}
\usepackage{cleveref}

\usetikzlibrary{shapes,arrows}
\usetikzlibrary{decorations.pathreplacing,angles,quotes}
\tikzstyle{block1} = [rectangle, draw, thick,fill=blue!10, text width=4.5em, text centered, rounded corners, minimum height=2em, minimum width = 5cm]

\tikzstyle{block2} = [rectangle, draw, thick,fill=blue!10, text width=2em, text centered, rounded corners, minimum height=2em]

\tikzstyle{line} = [draw, -latex']
\newcommand{\arvind}[1]{\todo[color=blue!10,size=\tiny]{#1}}

\newcommand{\norm}[1]{\left\lVert#1\right\rVert}

\ifpdf
  \DeclareGraphicsExtensions{.eps,.pdf,.png,.jpg}
\else
  \DeclareGraphicsExtensions{.eps}
\fi

\newsiamremark{remark}{Remark}
\newsiamremark{hypothesis}{Hypothesis}
\newsiamthm{claim}{Claim}
\newcommand{\mcA}{{\mathcal{A}}}

\newcommand{\mcK}{{\mathcal{K}}}
\newcommand{\mcD}{{\mathcal{D}}}
\newcommand{\mcP}{{\mathcal{P}}}
\newcommand{\ikrep}{{(k)}}

\newcommand{\Ff}{\mathcal{F}}

\newcommand{\Gg}{\mathcal{G}}

\newcommand{\PP}{\mathbb{P}}
\newcommand{\EE}{\mathbb{E}}
\newcommand{\II}{\mathbb{I}}

\newcommand{\N}{\mathbb{N}}
\newcommand{\ik}{{i}}

\newcommand{\mfT}{\mathfrak{T}}
\newcommand{\mfN}{{\mathfrak{N}}}
\newcommand{\tT}{{t\in\mfT}}

\newcommand{\tgamma}{{\tilde{\gamma}}}

\newcommand{\Id}{{\mathbb{1}}}
\newcommand{\R}{{\mathbb{R}}}
\renewcommand{\arvind}[1]{{}}

\newtheorem{assumption}{Assumption}

\graphicspath{{figures/}}

\headers{MFG in SREC Markets}
{A. Shrivats, D. Firoozi, and S. Jaimungal}

\title{A Mean-Field Game Approach to Equilibrium Pricing in Solar Renewable Energy Certificate Markets\thanks{SJ would like to acknowledge the support of the Natural Sciences and Engineering Research Council of Canada
(NSERC), [funding reference numbers RGPIN-2018-05705 and RGPAS-2018-522715]}}

\author{Arvind Shrivats\thanks{Operations Research \& Financial Engineering, Princeton University, Princeton, NJ, United States (\email{shrivats@princeton.edu}})
\and Dena Firoozi\thanks{Department of Decision Sciences, HEC Montréal, Montreal, QC, Canada (\email{dena.firoozi@hec.ca}}), 
\and Sebastian Jaimungal\thanks{Department of Statistical Sciences, University of Toronto, Toronto, ON, Canada (\email{sebastian.jaimungal@utoronto.ca}})}

\ifpdf
\hypersetup{
  pdftitle={Optimal Behaviour of Regulated Firms in Solar Renewable Energy Certificate (SREC) Markets},
  pdfauthor={Arvind Shrivats, Dena Firoozi, Sebastian Jaimungal}
}
\fi

\begin{document}

\maketitle
\begin{abstract}
Solar Renewable Energy Certificate (SREC) markets are a market-based system that incentivizes solar energy generation. A regulatory body imposes a lower bound on the amount of energy each regulated firm must generate via solar means, providing them with a tradeable certificate for each MWh generated. 
Firms seek to navigate the market optimally by modulating their SREC generation and trading rates. As such, the SREC market can be viewed as a stochastic game, where agents interact through the SREC price. We study this stochastic game by solving the mean-field game (MFG) limit with sub-populations of heterogeneous agents. Market participants optimize costs accounting for trading frictions, cost of generation, non-linear non-compliance costs, and generation uncertainty. Moreover, we endogenize SREC price through market clearing. 
We characterize firms' optimal controls as the solution of McKean-Vlasov (MV) FBSDEs and determine the equilibrium SREC price. We establish the existence and uniqueness of a solution to this MV-FBSDE, and prove that the MFG strategies form an $\epsilon$-Nash equilibrium for the finite player game. Finally, we develop a numerical scheme for solving the MV-FBSDEs and conduct a simulation study.
\end{abstract}



\section{Introduction} \label{sec:introduction}

Climate change has emerged as the preeminent global issue of the 21st century, resulting (belatedly) in the increased adoption of policy solutions to promote clean energy generation, or to disincentivize the production of energy via means that produce harmful emissions. A carbon tax is the most popular and well known policy solution which attempts to do the latter.

Among the set of policy implementations aimed at combating climate change are the class of so called market-based solutions. The most well-known of these policies are carbon cap and trade (C\&T) markets (also colloquially referred to as emissions markets). An alternative approach is known as `renewable energy certificate' (REC) markets. In these markets, a regulator imposes a lower bound on the proportion of energy generated from renewable sources that each regulated load serving entity (LSE) sells to the electrical grid, over a given time period, called a compliance period. Any generator (who may be a regulated LSE, but is not necessarily) is granted a certificate for each MWh of energy it generates from said renewable sources. Each regulated LSE must submit certificates in the amount of the lower bound prescribed to it by the regulator, while subject to a monetary penalty for failing to submit the required certificates. This penalty is typically per unit of non-compliance. That is, falling $m$ certificates short of the requirement results in a penalty $m$ times more severe than falling one certificate short of the requirement. In practice, these markets consist of many consecutive, non-overlapping compliance periods. 

The certificates that generators are granted are tradable. As such, a regulated LSE need not necessarily generate their certificates themselves. Instead, they may purchase the required certificates for compliance on the marketplace. Regulated firms therefore must choose between producing electricity from renewable means themselves, purchasing the certificates on the market, or a mix of both. Naturally, firms interact with one another through their behaviours in the REC marketplace and thus each regulated LSE has an impact on every other regulated LSE. 

REC markets may also be targeted to a specific type of renewable energy, in order to encourage growth of that particular energy type. The most common of the targeted REC markets are Solar REC (SREC) markets, which have been implemented in many areas in the northeastern United States and Europe (most notably, Sweden and Norway). We primarily focus on SREC markets in this work. In particular, we aim to understand how regulated LSEs should behave in SREC markets, with respect to their mix of SREC generation and trading activity such that they navigate the SREC market at minimum cost. For clarity, we remark here that we often refer to the regulated LSEs of an SREC market as regulated agents or firms. In order to model this in a complete manner, we must consider the potential interactions between agents through their marketplace actions, as well as the possibly heterogeneous nature of the regulated LSEs. 

The existing literature on SREC markets focuses heavily on the price formation of certificates. \cite{coulon_khazaei_powell_2015} proposes a stochastic model for economy-wide SREC generation, calibrating it to the New Jersey SREC market, and ultimately solves for the equilibrium SREC price. They further investigate the role of regulatory parameters and discuss potential takeaways for the efficient design of SREC markets. Other works, most notably \cite{amundsen2006price} and \cite{hustveit2017tradable}, discuss the volatility inherent in REC prices and analyze the price dynamics of these systems and the reasons for volatility. In \cite{khazaei2017adapt}, the authors go one step further, proposing an alternative SREC market design in order to stabilize SREC prices. \cite{shrivats2019behaving} flips the typical focus of works in SREC markets, using an exogenous price process as an input to model optimal behaviour on the part of a regulated agent in an SREC market, conducting numerical simulation studies to understand and characterize the nature of the firm's optimal controls.

There is considerably more literature on carbon C\&T markets, which REC markets resemble. In \cite{seifert_uhrig-homburg_wagner_2008}, the authors represent firm behaviour as the solution to an optimal control problem from the perspective of a central planner overseeing the emissions market in a single-period C\&T market and aiming to optimize total expected societal cost. The authors then characterize and solve for the carbon allowance price process. This is further extended to a multi-period C\&T market in \cite{hitzemann2018equilibrium}, which leads to the equilibrium carbon allowance price being expressed as a strip of European binary options written on economy-wide emissions. Agents' optimal strategies and properties of allowance prices are also studied by \cite{carmona2010market} and \cite{carmona_fehr_hinz_2009} via functional analysis arguments, within a single compliance period setup. Both works  make significant contributions through detailed quantitative analyses of potential shortcomings of these markets and their alternatives (in \cite{carmona2010market}) and of the certificate price and its properties (in \cite{carmona_fehr_hinz_2009}). In each of these works, the authors argue for the equivalence between the solution to the optimization problem for the central planner and an equilibrium solution whereby each individual agent optimizes their profit. These works do not, however, incorporate trading frictions into their model. Additionally, they focus on the implications of the solution to the optimization problem posed on the certificate price, as opposed to the nature of the optimal behaviour of the regulated agents themselves, which is the primary focus of this work. There are also notable works on structural models for financial instruments in emissions markets, such as \cite{howison_schwarz_2012} and \cite{carmona_coulon_schwarz_2012}.

In this work, we aim to understand how a regulated LSE should navigate an SREC market. As such, we formulate a single-period SREC market as an $N$-player stochastic game, with each player representing a regulated LSE who controls their planned SREC generation and SREC trading behaviour, aiming to navigate the SREC system at minimum cost. SREC generation is subject to some level of randomness, to represent the inherent variation in weather that impacts ones ability to generate solar energy. Moreover, we do not model SREC prices exogenously, rather, we  derive prices endogenously through enforcing market clearing.
Rather than solving the $N$-player stochastic game directly, we turn towards a mean-field game (MFG) approach. This approach aims to approximate the optimal behaviours of large populations of interacting agents in a stochastic game, making it a natural fit to apply to SREC markets. 

We use methods from variational analysis (similar to \cite{FCJ-Convex2018}) to establish existence and uniqueness of a solution to an arbitrary firm's optimization problem, given a particular mean-field distribution. We then fully characterize the optimal controls of a regulated agent and the corresponding mean-field distribution through a McKean-Vlasov FBSDE, for which we prove the existence and uniqueness of solutions. We further determine the equilibrium SREC price through the clearing condition -- which introduces a number of complexities to the MFG -- and prove that the MFG strategies have the $\epsilon$-Nash property for the finite player game. Finally, we develop a numerical scheme for solving the resulting MV-FBSDE and carry out numerical analyses  to better understand the nature of the solutions and their associated implications on certificate prices. 

MFGs themselves are very well-studied. They were originally developed in the works \cite{huang2007large,huang2006large}, and \cite{Lasry2006a, Lasry2006b, lasry2007mean}. Many extensions and generalizations of mean-field games and their applications exist. Among them include the probabilistic approach to MFGs and MFGs with common noise and master equation, detailed extensively in \cite{carmona2018probabilistic, carmona2013probabilistic, MasterEq2019, BensoussanBook2013}. 

In addition to these seminal works introducing the idea of MFGs and approaches to solving them, the numerical analysis of MFGs has also been well-studied. In particular, \cite{angiuli2019cemracs} and \cite{chassagneux2019numerical} study numerical approaches to the solutions of McKean-Vlasov forward-backward stochastic differential equations (MV-FBSDEs), which often characterize the solutions to MFGs. \cite{aziz2016mean} proposes a computational methodology to solve a set of forward-backward PDEs related to a MFG formulation of cellular communication networks.


MFGs have found numerous applications in engineering (\cite{aziz2016mean, KIZILKALE2019,Tembine2017}), economics (\cite{Gomes2014,Gomes2016}), and  in particular mathematical finance including optimal execution problems and portfolio trading (\cite{CarmonaLacker2015, Mojtaba2015, ThesisDena2019, FirooziISDG2017, FirooziPakniyatCainesCDC2017, Cardaliaguet2018,Lehalle2019, casgrain2018mean, Horst2018, David-Yuri2020}), systemic risk (\cite{CarmonaSysRisk2015,CarmonaSysRisk2018,JaimungalSysRisk2017}), commodities markets (\cite{aid2017coordination, Mouzouni2019,Sircar2017, brown2017oil, ludkovski2017mean}), and even cryptocurrencies (\cite{li2019mean}) -- just to name a few important contributions.

There has also been recent development in the field of MFGs with price formation. In particular, the contemporaneous works of \cite{fujii2020mean} and \cite{gomes2020mean} have added to this area. The former studies a generic asset pricing problem, endogenously pricing the asset by enforcing a trading balance  among the agents in the market. This includes the analysis of MFGs with common noise, and MFGs with sub-groups of agents. The authors characterize the solution to this problem as the solution to a MV-FBSDE. The latter presents an MFG price formation model where the supply for the commodity is a given deterministic function, and the clearing condition between supply and demand results in the interpretation of price as a Lagrange multiplier. The stochastic case is being further expanded upon in a follow-up (\cite{gomes2020mean2}) for linear-quadratic games, with the commodity price arising as the solution to a stochastic differential equation whose coefficients depend on the solution of a system of ordinary
differential equations.  We discuss the positioning of our work relative to these contemporaries below.

We first discuss our work relative to the more general literature on MFGs and emissions markets. Here, there are several key differences between our work and the extant literature. To the authors' knowledge, this is the first work applying MFG theory to SREC or C\&T markets in order to understand the dynamics of the certificate price process, accounting for trading frictions and generation costs, as well as the optimal behaviours of agents regulated by these markets. We believe insights into the features of these markets are of interest for regulatory bodies and regulated firms alike. As previously mentioned, this work separates itself from much of the analogous work on C\&T markets through its focus on regulated firms' optimal behaviour in addition to the certificate price process; the former is not a focus of much of the prior work in this field. As well, to the authors' knowledge, the extant literature does not incorporate trading frictions into their modelling of the costs firms face. We, however, explicitly account for trading frictions, reflecting the fact that trading is not costless in these markets. We also endogenize the certificate price process through a market clearing condition, allowing us to model price with no explicit assumptions on its drivers. Finally, using the MFG methodology  allows us to capture the interactions between agents arising through their trading and generation behaviour, and solve each agent's individual optimization problem directly, rather than solving a global problem from the perspective of a central planner aiming to minimize overall economy-wide costs. 

From the theoretical point of view, while the main reference we use is (\cite{carmona2013probabilistic}), we note that the problem we consider does not directly fit into their framework. This is due to a market clearing condition we enforce, as well as a cross dependence between the mean field distribution and the controls that we permit in our running costs (in particular the assumption in \cite[Prop. 3.7.]{carmona2013probabilistic}) does not hold for our problem.). 
Given the recency of \cite{fujii2020mean} and \cite{gomes2020mean}, as well as the similar subject matter, it is worth explicitly differentiating our contributions. While these works also involve MFGs with clearing conditions to endogenize price, our focus is on the specific application of MFG theory to SREC markets, as opposed to a generic asset, as in \cite{fujii2020mean}, or a commodity, as in \cite{gomes2020mean}. This results in notable differences in the models between these works and the contents of the paper. \cite{gomes2020mean} focuses largely on linear-quadratic problems, while, due to the nature of non-compliance penalties in SREC markets and the nonlinear dependence of the SREC price on the mean field distribution, our problem is decidedly not LQG. In their model, the supply of the commodity being priced is deterministic. In this work, however, the supply of the asset being priced is the controlled process itself, which is stochastic. While \cite{fujii2020mean} includes price formation in MFGs, as we do, we also provide a proof that fixed-point problem for the measure-flow solves the MFG, as well as existence and uniqueness proofs,  prove the $\epsilon$-Nash property for the finite-player game, and provide a numerical scheme for implementing the resulting MV-FBSDE. In addition, we carry out a detailed simulation study to understand the key drivers of regulated firm behaviour in the SREC market.
The remainder of the paper is structured as follows. In Section \ref{sec:model}, we formalize the SREC market and present our model for the regulated firms. Section \ref{sec:optimization_problem} formally poses the optimization problem that each agent faces in the $N$-player stochastic game. We introduce the MFG limit of this problem and derive the form of the optimal controls given a particular mean-field distribution in Section \ref{sec:mfg_problem}, using techniques from variational analysis. Further, we  characterize the MFG solution as a function of the solution to an MV-FBSDE, and prove existence and uniqueness of the solutions. In Section \ref{sec:epsilon-nash}, we relate the solution of the MFG to the solution of the $N$-player stochastic game and prove the $\epsilon$-Nash property. Finally, in Section \ref{sec:experiments}, we perform numerical simulations and discuss the key takeaways of the structure of the optimal controls across agents, as well as the implications these controls have on the certificate price and economy-wide compliance probabilities. \Cref{appendix:proofs} and \Cref{appendix:algorithm} contain, respectively,  the proofs of propositions and the details of the numerical implementation for solving the MV-FBSDE.

\section{The Model} \label{sec:model}

\subsection{SREC Market Rules and Assumptions}

We assume the following rules for the SREC market, which are exogenously specified and fixed (and identical to the rules assumed in \cite{shrivats2019behaving}). In an $n$-period framework, a firm must submit $(R_1, ..., R_n)$ SRECs at the end of the compliance periods $[0, T_1], ..., [T_{n-1}, T_n]$, respectively. For every SREC below $R_i$ at $T_i$ that a firm submits, they must pay $P_i$.

Firms receive SRECs through the generation of electricity through solar means. One SREC corresponds to one MWh of electricity produced via solar energy. A firm may also purchase or sell SRECs on the market. Throughout this work, we assume firms may bank leftover SRECs (that are unused for compliance) indefinitely. This is a simplification of reality -- in most cases, there are limits on how many times an SREC can be banked. However, this simplification allows us to greatly reduce the dimensionality of our state space, allowing the problem to be computationally tractable. After $T_n$, all firms forfeit any remaining SRECs. $T_n$ can be thought of as `the end of the world' -- there are no costs associated with any time after this. In the New Jersey SREC market, the largest and most developed in North America, the current penalty for non-compliance is \$$258$ per unit of lacking SREC, with the current SREC price slightly below that, at roughly \$$230$ per SREC, as of February 13, 2020, from \cite{SRECTrade}. For other practical SREC market details, we refer the interested reader to \cite{coulon_khazaei_powell_2015} and \cite{shrivats2019behaving}.

In this work, we focus heavily on a single-period SREC market framework. In this setting, the rules above apply with $n = 1$. Naturally, there is no concept of banking unused SRECs in a single-period framework. For notational convenience, we remove the subscripts in the quantities described above when discussing the single-period SREC market. That is, the regulated firm is required to submit $R$ 
SRECs at time $T$, representing their required production for the compliance period $[0, T]$. A penalty $P$ is imposed for each missing SREC at time $T$. There are assumed to be no costs after time $T$.

\subsection{Model Setup and Motivation}

In this subsection, we formulate the SREC market detailed above as a finite-player stochastic game, where agents are simultaneously striving to achieve maximum profit while interacting with one another through their SREC trading activities. We define the SREC price endogenously; that is, the SREC price is a function of the supply and demand for SRECs across the regulated firms. We allow for agent heterogeneity by defining a finite number of sub-populations of agents with unique behavioural parameters.

We work on the filtered probability space $(\Omega, \mathcal{F}, (\Ff_t)_{t \in \mfT}, \PP)$. All processes 
are assumed to be $\Ff$-adapted
unless otherwise stated. The filtration is defined further later in this section. 
We assume there are a finite number ($N$)  of firms and index them by $i \in \mfN := \{1, ..., N\}$. We further specify that each firm belongs to a sub-population (or class) 
and
index  the classes by $k \in \mcK := \{1, ..., K\}$, with $K \leq N$. All agents within a sub-population are assumed to be homogeneous and interchangeable, but are heterogeneous across classes. We also define
\begin{equation}
\mfN_k := \big\{ i \in \mathfrak{N}\; |\; \text{ agent $i$ belongs to sub-population $k$}\big\},
\end{equation}
and $N_k := |\mfN_k\rvert$ to be the number of firms in class $k$ for all $k \in \mcK$, with $\sum_{k \in \mcK} N_k = N$.
\begin{assumption}\label{ass:proportion}
The proportion of the total population of agents belonging to each class $k$ converges to a constant as the number of firms ($N$) increases. That is,
\begin{equation}
\lim_{N \rightarrow \infty} \tfrac{N_k}{N} = \pi_k \in (0, 1) \quad  a.s., \; \forall k\in\mcK \label{eq:proportion}
\end{equation} 
\end{assumption}

We allow the requirement $R$ to vary by sub-population, and refer to it henceforth by $R^k$, $k \in \mcK$. Agent $i$ (belonging to arbitrary sub-population $k$) can control their planned generation rate (SRECs/year) $(g_t^{i})_{\tT}$ (where $\mfT:=[0,T]$) at any given time  and their trading rate (SRECs/year) $(\Gamma_t^{i})_{\tT}$ at any given time. The processes $g^{i}$ and  $\Gamma^{i}$ constitute the control actions of firm $i$. Trading rates may be positive or negative, reflecting that firms can either buy or sell SRECs at the prevailing market rate for SRECs, respectively. Generation is assumed to be non-negative. We denote the collection of generation rates and trading rates by $\bm{g_t} \coloneqq (g_t^1, \cdots, ,g_t^N)$ and $\bm{\Gamma_t} \coloneqq (\Gamma_t^1, \cdots, \Gamma_t^N)$, respectively.

In an arbitrary time period $[t_1, t_2]$, the firm aims to generate $G^{i}_{t_1,t_2}:=\int_{t_1}^{t_2} g_t^i \,dt$, but in fact generates $G^{i, (r)}_{t_1,t_2} := \int_{t_1}^{t_2} g_t^i dt + \int_{t_1}^{t_2} \sigma_t^k \,dW_t^i$, where $\sigma_t^k$ is a deterministic function of time. The diffusive term $\sigma_t^k dW_t^{i}$ may be interpreted as the generation rate uncertainty at $t$ due to, e.g., stochasticity in the intensity of sunlight. We assume that a firm has a deterministic baseline  generation level $h_t^k$ (SRECs/year) and generating at a rate different from the baseline incurs a penalty.
For example, a firm may choose to rent solar panels from another generator in order to increase their SREC generation, which would carry some cost. Methods similar to \cite{coulon_khazaei_powell_2015} may be used to estimate $h_t^k$ and $\sigma_t^k$. We assume that $0 \leq h_t^k < \infty$ for all $t$. 

The firms controlled inventory (with controls $g_t^i, \Gamma_t^i$) is denoted by $X^{i, [N]}=(X_t^{i, [N]})_{t\in\mfT}$, and satisfies the SDE
\begin{equation}
dX_t^{i, [N]} = (g_t^i + \Gamma_t^i)\, dt + \sigma_t^k dW_t^i, \hspace{5mm}\text{ $\forall i \in \mfN_k$} \label{eq:state_dynamics}
\end{equation}
where $W = \{W^i = (W_t^i)_{t \in \mathfrak{T}}, i \in \mathfrak{N}\}$ is a set of N independent standard Brownian motions,
and $W^i$ is progressively measurable with respect to
the filtration $\mathcal{F}^W \coloneqq (\mathcal{F}_t^W)_{t\in \mfT} \subset \mathcal{F}$ generated by $W$. Each firm has an initial inventory of SRECs, given by the collection of random variables $\{\xi^i\}_{i = 1}^N$. Naturally, these choices imply that SRECs are fractional and infinitely divisible in our model.
Moreover the empirical distribution $\mu^{[N]}_t$ of $X = \{X^{i, [N]} = (X_t^{i, [N]})_{t \in \mathfrak{T}}, i\in \mathfrak{N}\}$ is defined as
\begin{equation}
    \mu^{[N]}_t(dx) = \frac{1}{N} \sum_{i=1}^N \delta_{X_t^{i, [N]}}(dx), 
\end{equation}
where $\delta_y(\cdot)$ denotes the Dirac measure with unit point mass at $y$. 
\begin{assumption} \label{IntialStateAss}
The initial states $\{\xi^i\}_{i = 1}^N$ 
are identically distributed, mutually independent, and independent of $\mathcal{F}^W$. Moreover, $\sup_{i} \mathbb{E}[\Vert \xi^i\Vert^2] \leq c < \infty $, $i \in\mfN$, with $c$ independent of $N$. 
\end{assumption} 


We denote the equilibrium SREC price process by  $S^{\mu^{[N]}}=(S_t^{\mu^{[N]}})_{t\in\mfT}$. With this notation, we emphasize the dependence of the SREC price on the controlled states of all agents; indeed, the SREC market (and accordingly, the equilibrium SREC price) is the mechanism by which agents interact with one another. We do not make any explicit assumptions about the dynamics of $S_t^{\mu^{[N]}}$.
Instead, we derive its form endogenously in \Cref{sec:SREC_price} through the clearing condition, i.e. that trading must be zero-sum among all agents. 

We next denote by $\Gg^i= (\Gg_t^i)_{t \in \mathfrak{T}}$ the filtration that an individual adapts their strategy to, and is defined as  
\begin{equation}\label{individualFilt}
    \Gg_t^i := \sigma\left((X_u^{i, [N]})_{u \in [0, t]}\right)  \vee \sigma\left(\left(S_u^{\mu^{[N]}}\right)_{u \in [0, t]}\right),%
\end{equation}
which is the sigma algebra generated by the $i$-th firm's SREC inventory path and the SREC price path. Note we assume that all firms have knowledge of the the initial distribution (but not the actual value) of other firms' SRECs; that is, firms have knowledge of $\mu_0^{(k)}$, $\forall \;k\in\mcK$. 
The full filtration $\Ff=(\Ff_t)_{t\ge0}$ is defined as $\Ff_t=\bigvee_{i\in\mfN} \Gg_t^i$. 
We define the set of square integrable controls
\begin{equation}
    \mathbb{H}_t^2 := \left\{\left. (g, \Gamma) : \Omega \times \mathfrak{T} \rightarrow \R^2 \; \right|\; \EE\left[{\textstyle\int_0^T} [(g_t)^2+(\Gamma_t)^2]\, dt \right] < \infty\right\}.
\end{equation}

\begin{assumption} \label{ass: MinorContrAction} The set of admissible controls for firm $i\in\mfN$  is
\begin{align}
    \mcA^i := \left\{ (g, \Gamma) \in \mathbb{H}_t^2\,\, \text{s.t.} \,\, g_t \geq 0\, \text{ for all } \, t \in \mathfrak{T}\,\, \text{and} \,\, (g, \Gamma) \text { $\Gg^i$-adapted}\,\, \right\}. \label{firmsAdmissibleCntrl}
\end{align}
\end{assumption}
This is a closed and convex set. As an individual firm cannot observe another firms' SREC inventories, the restriction in the set above is to decentralized $\Gg^i$-adapted strategies. We finally denote the set of the admissible  strategies for all other agents by
\[
\mcA^{-i} := \bigtimes_{j \in \mathfrak{N},\, j \neq i} \mcA^j, \qquad \forall\,i\in\mfN.
\]

\section{Constructing the Agents Optimization Problem} \label{sec:optimization_problem}

Each agent chooses their controls with the goal of minimizing the cost (equivalently, maximizing the profit) they incur through the SREC market over the time span $[0, T]$. 
Specifically, the $i$-th firm (belonging to sub-population $k$, i.e., $i\in\mfN_k$) aims to minimize the cost functional $J^{i}: \mcA^i \rightarrow \R$ where
\begin{equation}
J^{i} (\bm{g}, \bm{\Gamma})= \EE\left[ \int_0^T \left(C^k(g_u^i, h_u^k) + H^k(\Gamma_u^i) + S_u^{\mu^{[N]}} \,\Gamma_u^i \right) du + P(R^k - X_T^{i, [N]})_+\right].
\label{eq:agentCostFn}
\end{equation} 


The agent's objective comprises of four distinct terms. The first corresponds to costs associated with planned SREC generation. Specifically, the agent incurs the cost $C^k(g, h)$ per unit time for its level of planned SREC generation. We choose 
\[
C^k(g,h) := \tfrac{1}{2}\zeta^k (g - h)^2.
\]
 A related problem is studied in \cite{aid2017coordination} in the context of expanding solar capacity where costs are quadratic. This is both differentiable and convex, both of which are desirable properties for our analysis. Our cost is best interpreted as a firm renting solar capacity to increase their planned SREC generation rate, as opposed to an investment cost where they would increase their long-run SREC generation capacity. 

In practice, when firms choose to expand their solar energy generation, they often do so by undertaking projects to build or acquire solar energy generation facilities, which they then operate. This increases their ability to generate solar energy (and thus SRECs) on an ongoing basis. In our model, a firms planned generation choice does not have a long-term impact on their baseline generation rate, as that extension is left for future work.


The second term corresponds to a trading speed penalty. The firm incurs a trading penalty of $H^k(\Gamma)$, per unit time. This induces a constraint on their trading speed. Specifically, we choose 
\[
H^k(\Gamma) = \tfrac{1}{2}\gamma^k (\Gamma)^2, \qquad \gamma^k > 0.
\]
This cost introduces  a key difference between our model and the extant literature in the C\&T work. As mentioned in Section \ref{sec:introduction}, prior work in this field do not incorporate trading frictions into their model (see \cite{hitzemann2018equilibrium}, \cite{seifert_uhrig-homburg_wagner_2008}, \cite{carmona_fehr_hinz_2009}, and \cite{carmona2010market}). 

The third term corresponds to the cost (revenue) generated when purchasing (selling) an SREC on the market, with the firm paying (receiving) the equilibrium SREC price $S_t^{\mu^{[N]}}$. 

The fourth and final term corresponds to the non-compliance penalty the firm faces if they fail to submit the required number of SRECs at the end of the compliance period. We allow the requirement to vary based on sub-population. In doing so, we can also incorporate the participants in SREC markets who do not have an RPS obligation at all, but do possess the means to generate SRECs, which they may sell on the market. We observe that the non-compliance penalty is not differentiable at $X_T^{i,[N]} = R^k$.  To avoid technical issues, we introduce a `regularized' version of the non-compliance penalty.
We require this regularized version of the non-compliance penalty to be convex, once differentiable, and twice differentiable a.e.. Requiring the former maintains the convexity of the optimization problem, and requiring the latter ensures our functional is G\^ateaux differentiable everywhere. There are many suitable regularizations; one such example is
\begin{equation}
F_\delta(x) = \begin{cases}
0, &\text{$x < -\delta$}, 
\\
\tfrac{(x+\delta)^2}{4\delta}, &\text{$|x| \leq \delta$}, 
\\
x, &\text{$x  > \delta$},
\\
\end{cases} \label{eq:regularizedNCC}
\end{equation}
where $\delta > 0$. It is obvious that $F_\delta(R^k - X_T^{i, [N]})$ and $(R^k - X_T^{i, [N]})_+$ agree everywhere except $|x| < \delta$. As $\delta \rightarrow 0$, the length of the domain on which $F_\delta(R^k - X_T^{i, [N]})$ disagrees with $(R^k - X_T^{i, [N]})_+$ tends to $0$. This choice is clearly convex, once differentiable everywhere, and almost everywhere twice differentiable. 

With the specific form of the various cost, the performance criterion \eqref{eq:agentCostFn}  becomes
\begin{equation}
J^{i} (\bm{g}, \bm{\Gamma}) = \EE\left[ \int_0^T \left(\tfrac{\zeta^k}{2} (g_u^i - h_u^k)^2 + \tfrac{\gamma^k}{2} (\Gamma_u^i)^2 + S_u^{\mu^{[N]}} \Gamma_u^i \right) du + P\, F_{\delta} (R^k - X_T^{i, [N]})\right], \label{eq:agentCostFnRegularized}
\end{equation} 
for agents $i\in\mfN_k$. 

To summarize the key notation introduced in the current and previous section for the readers benefit, we present Table \ref{tbl:Notation} below.
\begin{table}[bt]
\centering
\begin{threeparttable}
 \begin{tabular}{lll}
 \hline
 Quantity & Description & Units \\
   \hline
  $g_t^{i}$ & Planned generation rate for $i$-th agent & SRECs / year \\
 $\Gamma_t^{i}$ & Trading rate for $i$-th agent & SRECs / year\\
 $\zeta^k$ & Generation cost parameter for class $k$ firms & \$ / SRECs$^2$ \\
 $\gamma^k$ & Trading speed penalty parameter for class $k$ firms & \$ / SRECs$^2$ \\
 $\sigma_t^k$ & Volatility function for class $k$ firms & SRECs / $\sqrt{\text{year}}$ \\
 $h_t^k$ & Baseline generation rate for class $k$ firms & SRECs / year \\
 $S_t^{\mu^{[N]}}$ & SREC market price at time $t$& \$ / SREC \\
 $X_t^{i}$ & SREC inventory of $i$-th agent at time $t$ & SRECs \\
 $T$ & Terminal time of period & Years \\
 $P$ & Penalty per unit of non-compliance & \$ / lacking SREC \\
 $R^k$ & SREC Requirement for class $k$ firms & SRECs \\
\hline
\end{tabular}
\end{threeparttable}
\caption{Notation}
\label{tbl:Notation}
\end{table}

Agents within a sub-population have the same cost parameters, and consequently,  those agents act, in equilibrium, in a similar manner. Each individual agent's strategy, however, is adapted to their own inventory and as such, agents' strategies are not identical, even within the same sub-population.

%

As previously stated, all agents seek to minimize their own costs, and we look for the optimal control for all agents simultaneously. Specifically, we seek the Nash equilibrium -- a collection of controls $\{(g^{i, \star}, \Gamma^{i,\star}) \in \mcA^i: i \in \mathfrak{N}\}$ such that

\begin{equation}
(g^{i,\star}, \Gamma^{i,\star}) = \underset{(g^i, \Gamma^i) \in \mcA^i}{\arg \inf} \, J^{i}(g^i, \Gamma^i, \bm{g}^{-i, \star}, \bm{\Gamma}^{-i, \star}), \hspace{5mm} \forall i \in \mfN_k, \, \forall k \in \mcK, \label{eq:problem_statement}
\end{equation}
where  $(g^i, \Gamma^i, \bm{g}^{-i, \star}, \bm{\Gamma}^{-i, \star})$ denotes the set of strategies $(\bm{g}^\star, \bm{\Gamma}^\star)$ with $(g^{i, \star}, \Gamma^{i, \star})$  replaced by $(g^i, \Gamma^i)$. This means that an individual agent cannot benefit by unilaterally deviating from the Nash strategy. 

Additionally, we seek  $S_t^{\mu^{[N]}}$ such that
\begin{equation}
    \frac{1}{N}{\sum_{i \in \mathfrak{N}} \Gamma_t^{i, \star}} = \frac{1}{N}{\sum_{k\in\mcK} \sum_{i\in\mfN_k} \Gamma_t^{i, \star}} = 0, \label{eq:market_clearing_condition}
\end{equation} $t-a.e, \, t \in \mathfrak{T}, \, \PP$ a.s. In the finite-player setting, this condition states that trading is zero-sum among regulated firms in the SREC market, and is also known as the market clearing condition. This condition results in an equilibrium price $S_t^{\mu^{[N]}}$ endogenously defined through the optimal controls and market clearing.

As the cost functional for each agent is impacted by the actions of all other agents through the price process, and each agent has private information, solving \eqref{eq:problem_statement} for the finite player game is challenging. As such, in the next section, following the mean-field game methodology, we treat the problem in its infinite-player limit (as $N \rightarrow \infty$)  to render it tractable. The goal is to obtain a set of strategies that forms a Nash equilibrium in the infinite-player limit, and then apply that strategy to the finite-player case to achieve an approximate equilibrium, assuming the number of players is sufficiently large.

 
\section{The Mean-Field Game Limit} \label{sec:mfg_problem}

\subsection{Formulation of the Mean Field Game}
The stochastic game specified in the prior sections is difficult to solve directly. Instead, we focus our efforts on solving the stochastic game in the limit as $N \to \infty$ following the mean-field game (MFG) methodology. We assume the proportion of agents in each sub-population remains constant, and that all agents are minor agents (that is, in the limit, each agent is insignificant relative to the rest of the market). In the infinite-population limit, the finite player game becomes a MFG where agents interact through the population distribution of states, rather than directly. This makes the MFG simpler to solve than the finite player game, and the formulation of the MFG is the focus of this subsection. Proceeding this, we  derive the optimal controls and the equilibrium SREC price.

\textcolor{black}{In the infinite player setting, our notation deserves a brief discussion. We maintain our notation of indexing by firms, even in the infinite limit case. This allows us to be precise about what we are referring to, with no loss of clarity, as we will explicitly point out the sub-population each agent belongs to. Where appropriate, however, we use a superscript ${}^\ikrep$ rather than a superscript ${}^i$ to denote quantities relating to a representative (unspecified index) agent belonging to sub-population $k$. Parameters such as $\zeta, \gamma, h$ retain their superscript $k$ notation without parentheses, as they are real numbers and not a representative process.}



Prior to defining the firm's cost functional in the infinite-player limit, we make an additional definition.
\begin{definition} [Mean Field Distribution of States] \label{def:mean_field_distribution} 
In the infinite-player setting, we denote the mean-field distribution of SREC inventory for agents in sub-population $ k \in \mcK$ by $\mu_t^{(k)}$. Specifically, we introduce the flow of measures
\begin{equation}
    \mu^{(k)} = (\mu_t^{(k)})_{t \in\mfT},
    \qquad \mu_t^{(k)} \in \mcP(\R), \;\forall t\in\mfT
\end{equation}
for all $k \in \mcK$,  where $\mathcal{P}(\R)$ represents the space of probability measures on $\R$, such that $\mu_t^{(k)}(A)$ is the probability that a representative agent from 
sub-population $k$ has an SREC inventory belonging to the set $A \in \mathcal{B}(\R)$, at time $t$. 
Furthermore, we define
\begin{equation}
    \bm{\mu} = (\{\mu_t^{(k)}\}_{k \in \mcK})_{t \in \mfT}
\end{equation}
to be the flow of the collection of all mean-field measures.

\end{definition}


In \eqref{eq:agentCostFnRegularized}, only the equilibrium SREC price is dependent on the actions of other agents in the game. In the infinite-player limit, we replace $S_t^{\mu^{[N]}}$ with a process $S_t^{\bm{\mu}}$ which, in line with the mean-field limit, we assume is not impacted by any individual agent's behaviour, but is impacted by the mean-field distribution. We show in Section \ref{sec:SREC_price} that the equilibrium SREC price indeed has this independence. 

In the mean field setting, we define the clearing condition as the infinite-player limit of \eqref{eq:market_clearing_condition}
\begin{equation}
    \lim_{N \rightarrow \infty} \sum_{i \in \mfN} \frac{\Gamma_t^i}{N} = 0, \label{eq:infinite_player_clearing_condition}
\end{equation}
$ t-a.e.,\, \PP-a.s.$, for all $t \in \mfT$. This is analogous to the average trading rate (across agents) vanishing at all times in the limit.
\textcolor{black}{We could equivalently define this clearing condition in terms intrinsic to the mean field setting, rather than asymptotically. In doing so, the clearing condition is}
\begin{equation}
    \color{black} 
    \sum_{k \in \mcK} \pi_k \EE\big[\Gamma_t^{(k)}\big] = 0 \label{eq:infinite_player_clearing_condition_reviewer_suggested}
\end{equation}
\textcolor{black}{$ t-a.e.,\, \PP-a.s.$, for all $t \in \mfT$, where $\Gamma^{(k)}$ represents the trading rate of a generic (representative) agent in sub-population $k$.}

The dependence on the mean-field arises implicitly through $S_t^{\bm{\mu}}$. We postpone the characterization of this dependence until Section \ref{sec:SREC_price}, after we  the problem in the mean-field setting, and obtained the functional form of the optimal controls. We do this as $S_t^{\bm{\mu}}$ is ultimately derived from the market clearing condition of the average (optimal) trading being zero 
across all agents. Until such time, it suffices to consider $S_t^{\bm{\mu}}$ as a generic stochastic process that ultimately depends on the mean-field distribution in \cref{def:mean_field_distribution}. 
The SREC price $S_t^{\bm{\mu}}$ must be bounded above by $P$, as a firm would not pay above $P$ to obtain an SREC to avoid a penalty worth $P$.

We now denote the $i$-th agents' SREC inventory in the infinite-population limit by $X^\ik = (X_t^\ik)_{t \in \mfT}$. Subsequently, we denote firm-$i$'s cost functional, $i\in\mfN_k$, in the infinite-population limit  by $\overline{J}^{\ik}$ and equals

\begin{align}
\overline{J}^\ik (g^\ik, \Gamma^\ik; \bm{\mu}) = \EE\biggl[ & \int_0^T 
\left\{\begin{pmatrix}g_u^\ik & \Gamma_u^\ik \end{pmatrix} Q^k \begin{pmatrix}
g_u^\ik \\ \Gamma_u^\ik \end{pmatrix}+ \begin{pmatrix}g_u^\ik & \Gamma_u^\ik \end{pmatrix} \begin{pmatrix}
- \zeta^k h_u^k 
\\
S_u^{\bm{\mu}}
\end{pmatrix} + \tfrac{\zeta^k}{2} (h_u^k)^2
\right\}
du 
+ P\, F_{\delta}(R^k - X_T^\ik) \biggr], \label{eq:agentCostFnRegularizedMF}
\end{align}
where $Q^k = \tfrac{1}{2}\begin{pmatrix} \zeta^k & 0 \\ 0 & \gamma^k \end{pmatrix}$ is positive definite, and the individual agent adapts their strategy to the filtration $\Gg^\ik= (\Gg_t^\ik)_{t \in \mathfrak{T}}$ (with a slight abuse of notation, as this is the same notation used for the finite player game), where 
\begin{equation}\label{individualFilt_infpop}
    \Gg_t^\ik := \sigma\left((X_u^\ik)_{u \in [0, t]}\right) \vee \sigma\left(\left(S_u^{\bm{\mu}}\right)_{u \in [0, t]}\right)\,.
\end{equation}
This is the sigma algebra generated by the \textcolor{black}{individual} firm's SREC inventory and \textcolor{black}{market clearing} MFG SREC price. We maintain the assumption that all firms have knowledge of the the initial distribution (but not the value) of other firms' SRECs; that is, firm's have knowledge of $\mu_0^{(k)}$. The admissible set is now denoted by $\mcA^\ik$, retaining its definition in \eqref{firmsAdmissibleCntrl} but with the filtration $\Gg$ above.

As before, the state dynamics are
\begin{equation}
dX_t^\ik = (g_t^\ik + \Gamma_t^\ik)\, dt + \sigma_t^k \,dW_t^\ik, 
\label{eq:state_dynamicsMF}
\end{equation}
where $W^\ik = \{W^\ik = (W_t^\ik)_{t \in \mathfrak{T}}, i\in \mfN\}$ is a set of independent standard Brownian motions. Each individual firm has a random initial inventory of SRECs, given by the random variable $\xi^\ik.$ 

Similar to the finite-player game, we aim to find the set of strategies that form a Nash equilibrium \textcolor{black}{in the infinite player context}. That is, we search for a collection of controls across all agents $\{(g^{\ik,\star}, \Gamma^{\ik,\star}) \in \mcA^\ik \}_{i\in\mfN}$ such that 
\begin{equation}
(g^{\ik,\star}, \Gamma^{\ik,\star}) = \underset{(g, \Gamma) \in \mcA^\ik}{\arg \inf} \, \overline{J}^\ik(g,\Gamma; \bm{\mu}), \hspace{5mm} \forall i \in \N\,.  \label{eq:problem_statementMF}
\end{equation}
To solve \eqref{eq:problem_statementMF} and obtain the optimal controls, we use techniques from variational analysis (see \cite{FCJ-Convex2018} for a detailed variational analysis for a general class of LQG systems and LQG mean field game systems with major and minor agents). As with the finite-player case, we also seek $S_t^{\bm{\mu}}$ such that the clearing condition \eqref{eq:infinite_player_clearing_condition} holds. In the following subsection, we detail the approach that we take in solving \eqref{eq:problem_statementMF}. \textcolor{black}{In the MFG framework,  $\bm{\mu}$ is both an \textit{input} and an \textit{output} to the problem, as it determines the SREC price and the optimal controls  determine the distribution of the states through \eqref{eq:state_dynamics}. 
This results in a fixed-point problem on the space of measure flows which will be discussed in further detail in this section}.

\subsection{Deriving the optimal controls}\label{sec:variationalAnalysis}

Our approach to deriving the optimal controls is as follows. We first consider the problem of finding the optimal controls for a particular mean field distribution. That is, we treat an arbitrary $\bm{\mu}$ as an \textit{input} 
to the problem. We establish that the cost functional \eqref{eq:agentCostFnRegularizedMF} is strictly convex and differentiable (in the sense that its G\^ateaux derivative exists everywhere in all directions). From here, we find the value of the controls for which the G\^ateaux derivative vanishes (in all directions), in terms of a stochastic process $Y_t^i$ (an adjoint process, introduced later) as well as $\bm{\mu}$.
The convexity of the functional guarantees such  controls are optimal, and moreover, they correspond to a unique global optimum, given the mean-field $\bm{\mu}$. 
Subsequently, we use the form of the optimal controls to derive a system of MV-FBSDEs, and seek a set of mean-field distributions $\bm{\mu}$ such that the forward processes of the solution to the system of FBSDEs we propose has marginal distributions consistent with said mean-field distributions. This approach, summarized in \Cref{fig:solution_methodology}, \textcolor{black}{ results in a fixed point problem}. 
\tikzstyle{sec2} = [rectangle, draw, thick,fill=blue!10, text width=3em, text centered, rounded corners, minimum height=2em]
\tikzstyle{sec3} = [rectangle, draw, thick,fill=blue!10, text width=5.5em, text centered, rounded corners, minimum height=2em]
\tikzstyle{sec4} = [rectangle, draw, thick,fill=blue!10, text width=6em, text centered, rounded corners, minimum height=2em]
\begin{figure}[bt]
\centering
\begin{tikzpicture}[node distance = 3.75cm]
    \node [sec2] (input) {\footnotesize Take $\bm{\mu}$ \\ \footnotesize as input};
    \node [sec3, right of=input, xshift=0.3cm] (x0xbar) {\footnotesize Obtain $g_t^{\ik, \star}, \Gamma_t^{\ik, \star}$ \\ \tiny (in terms of $Y_t^\ik$ \& $\bm{\mu}$)};
    \node [sec3, right of=x0xbar, xshift=0.3cm] (x0hat) {\footnotesize   Find $S_t^{\bm{\mu}}$ \\ \tiny (in terms of $Y_t^\ik$ \& $\bm{\mu}$)};
    \node [sec4, right of=x0hat, xshift=0.5cm] (u0) {\footnotesize  Formulate FBSDE and solve for $X_t^\ik, Y_t^\ik, \mathcal{L}(X_t^\ik)$ \\ \footnotesize for all $k \in \mcK$};

    
    
    \path [line, thick] (input.0)([xshift=0cm]input.east) |- (x0xbar.180) node [xshift=-1.1cm,above, text width = 1.3cm, text centered] (TextNode) {\tiny Variational Analysis}; 
    
    \path [line, thick] (x0xbar.0)([xshift=0cm]x0xbar.east) |- (x0hat.180) node [xshift=-0.85cm,above, text width = 1.3cm, text centered] (TextNode) {\tiny Clearing condition} node [xshift=-1.75cm, yshift=0.5cm,above] (TextNode) {};

    \path [line, thick] (x0hat.0)([xshift=0cm]x0hat.east) |- (u0.180) ;



    \draw [dashed] (u0.270) -| ([yshift = -1cm]u0.270) -| node [yshift = 0cm,xshift = -0.0cm]{}(input.south) |- (input.270) node [xshift=6.15cm,yshift=-1.0cm, text centered] (TextNode) {\footnotesize Check $\mu_t^{(k)}, \mathcal{L}(X_t^i)$ for consistency, for all $k \in \mcK$};

 
\end{tikzpicture}
\caption{Structure of solving for optimal actions and clearing prices in the mean-field game limit.} \label{fig:solution_methodology}
\end{figure}
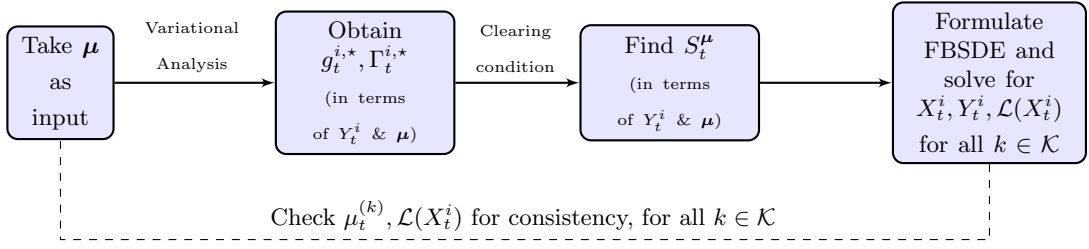

\begin{proposition}[Convexity of Cost Functional] \label{convexity}
For any given mean-field distribution $\bm{\mu}$, the functional $\overline{J}^\ik$ in \eqref{eq:agentCostFnRegularizedMF} is strictly convex in $\mcA^\ik$.
\end{proposition}
\begin{proof} See \Cref{Proof:prop:convexity}.
\end{proof}

We next establish  the G\^ateaux derivative of \eqref{eq:agentCostFnRegularizedMF} exists  and determine its explicit form.

\begin{proposition}[Existence and form of G\^ateaux derivative] \label{gateaux_deriv}
For any given mean-field distribution $\bm{\mu}$, the cost functional for an arbitrary agent $i$ in sub-population $k$, $\overline{J}^\ik(\cdot, \cdot; \bm{\mu})$, is everywhere G\^ateaux differentiable in $\mcA^\ik$, in all directions\textcolor{black}{, for all $k\in\mcK$.} The G\^ateaux derivative at $(g, \Gamma) \in \mcA^\ik$ in the direction of $(\omega^g, \omega^\Gamma)\in\mcA^\ik$ is
\begin{subequations}
\begin{align}
\langle \mcD \overline{J}^\ik(g^\ik, \Gamma^\ik; \bm{\mu}), \omega^g \rangle &= \EE\left[ \int_0^T \omega_t^g \left( \zeta^k (g_t^\ik - h_t^k) - P Y_t^
\ik \right)dt\right],\label{eq:GD_gen} \\
\langle \mcD \overline{J}^\ik (g^\ik, \Gamma^\ik; \bm{\mu}), \omega^\Gamma \rangle &= \EE\left[ \int_0^T \omega_t^\Gamma \left( \gamma^k \Gamma_t + S_t^{\bm{\mu}} - P Y_t^\ik\right)dt\right],\label{eq:GD_trade}
\end{align}
\end{subequations}
where 
\begin{equation}
Y_t^\ik :=  \EE \left[\left.F_{\delta}^\prime(R^k - X_T^\ik)\, \right|\, \Gg_t^\ik \right].
\label{eqn:defn-Y}
\end{equation}
\end{proposition}
\begin{proof} See \Cref{Proof:prop:gateaux_deriv}.
\end{proof}


Before continuing, we make an observation about $ Y_t^\ik$. $ Y_t^\ik$ is approximately the non-compliance probability for \textcolor{black}{a generic }
firm $i$ 
from sub-population $k$. To see this, recall that $F_\delta(x)$ is a convex and differentiable approximation of $a(x)=(x)_+$, and
\begin{equation*}
  a^\prime(x)= 
  \begin{cases}
   \Id_{x > 0}, & \text{for } x \neq 0 \\
    \text{d.n.e}, & \text{for } x = 0.
  \end{cases}
\end{equation*}
From the definition of $F_\delta(\cdot)$ in \eqref{eq:regularizedNCC}, we have that $F_\delta^\prime(x) \overset{\delta \rightarrow 0}{\longrightarrow} a^\prime(x)$ pointwise a.e., and for $\delta$ sufficiently small
\begin{equation}
    Y_t^\ik = \EE \left[F_{\delta}^\prime(R^k - X_T^\ik) | \Gg_t^\ik \right] \approx \EE \left[a^\prime(R^k - X_T^\ik) | \Gg_t^\ik \right] = \PP(X_T^\ik < R^k\, |\, \Gg_t^\ik). \nonumber
\end{equation}
Thus, we often refer to $Y_t^\ik$ as \textcolor{black}{the generic} agent's-$i$'s non-compliance probability at $t$. 

From Propositions \ref{convexity} and \ref{gateaux_deriv},   $\overline{J}^\ik$ is strictly convex and everywhere G\^ateaux differentiable in all directions. This implies that an element $\nu \in \mcA^\ik$ that makes the G\^ateaux derivative vanish for all directions is a minimizer of $\overline{J}^\ik$, and moreover, that  minimizer is unique. 
We now establish the conditions which are sufficient and necessary for the G\^ateaux derivative to vanish, which allows us to find solutions to \eqref{eq:problem_statementMF} and interpret the optimal controls in the context of SREC markets.

\begin{proposition}[Vanishing of G\^ateaux Derivative] \label{prop:optimality}
The G\^ateaux derivative \eqref{eq:GD_gen}-\eqref{eq:GD_trade} vanishes in all directions if and only if:
\begin{subequations}
\begin{align}
g_t^\ik - h_t^k - \tfrac{P}{\zeta^k} Y_t^\ik  &= 0,  \qquad\text{and}
\label{eq:gen_nec_cond} \\
\Gamma_t^\ik -  \tfrac{1}{\gamma^k} \left( P Y_t^\ik - S_t^{\bm{\mu}}\right) &= 0 .  \label{eq:trade_nec_cond}
\end{align}
\label{eq:actions}
\end{subequations}
\end{proposition}
\begin{proof} See \Cref{Proof:prop:optimality}. 
\end{proof}

These set of results allow us to determine the optimal controls in the mean-field game limit. 
\begin{proposition}[Optimal Controls] \label{prop:optimal_controls}
Given a mean-field flow $\bm{\mu}$, the collection of controls $\{g_t^{\ik, \star}, \Gamma_t^{\ik, \star}\}_{i \in \mfN}$ is a unique solution to \eqref{eq:problem_statementMF} if and only if for agent $i$ in sub-population $k$, $(g_t^{\ik, \star}, \Gamma_t^{\ik, \star}) \in \mcA^\ik$, and the G\^ateaux derivative vanishes for all $\omega \in \mcA^\ik$. \textcolor{black}{This holds for all $i \in \mfN_k$, for all $k\in\mcK$}. 

In particular, the optimal controls for the representative agent in sub-population $k$ are
\begin{align}
    g_t^{\ik, \star}  &= h_t^k + \tfrac{P}{\zeta^k} Y_t^\ik, \quad \text{and}
    \label{eq:optG} \\
\Gamma_t^{\ik, \star} &=   \tfrac{1}{\gamma^k} \left( P \,Y_t^\ik - S_t^{\bm{\mu}}\right).  \label{eq:optGamma}
\end{align}

\end{proposition}
\begin{proof}

The first part of this proposition is taken care of by \Cref{convexity} and Proposition 2.1 in \cite{ekeland1999convex}.

All that remains is to show that $g_t^{\ik, \star}$ and $\Gamma_t^{\ik, \star}$  in \eqref{eq:optG}-\eqref{eq:optGamma} are admissible. Observe that $h_t^k$ is deterministic, non-negative, and bounded. $ Y_t^\ik$ is, moreover, non-negative and bounded (by definition of $F_\delta$), as well as $ \Gg^\ik$-adapted (by construction). Therefore, $ g^{\ik, \star}$ is admissible.

Similarly, as $ Y^\ik$ and $S^{\bm{\mu}}$ are $ \Gg^\ik$-adapted and bounded,
\eqref{eq:optG}-\eqref{eq:optGamma} clearly satisfy  \Cref{prop:optimality}.
\end{proof}


The controls  in \eqref{eq:optG} and \eqref{eq:optGamma} provide us with insight about the nature of regulated firms optimal behaviour. $ P\,Y_t^\ik$ may  be interpreted as the expected non-compliance cost avoided by obtaining a marginal SREC. This can also be seen as the marginal benefit of holding an additional SREC.

With this interpretation in mind, we can restate \eqref{eq:optG} as the following:
\begin{equation}
    \zeta^k(g_t^{\ik,\star} - h_t^k) = P \, Y_t^\ik \,.
\end{equation}
This is equivalent to stating that the marginal cost of generation equals  the marginal benefit of generation. A similar result occurs in \cite{hitzemann2018equilibrium}. Isolating $g_t^{\ik,\star}$ implies that firms  generate at their baseline rate plus an amount proportional to the marginal benefit of holding an additional SREC. This makes intuitive sense, as it implies a firm that is in danger of non-compliance generates above their baseline in order to strive for compliance. Conversely, a firm that is assured of compliance  generates their baseline, which they obtain without incurring any cost. 

Similarly, we can interpret \eqref{eq:optGamma} as the optimal trading rate of a firm being proportional to the difference between the marginal value of holding an SREC and the marginal value of selling an SREC. This is scaled by the firm's cost of trading. Once again, this is an intuitive result. If SREC prices are high relative to the benefit a firm would receive from holding an SREC, a firm is more inclined to sell SRECs (and vice-versa). The degree to which this occurs is dependent on the difference between these two quantities, as well as the cost the firm incurs by partaking in the market ($\gamma^k$).

We emphasize  that the results in this section hold for a given, but arbitrary, mean field distribution $\bm{\mu}$. With this, we obtain the existence and uniqueness of the optimal trading and optimal planned generation of \textcolor{black}{any generic firm in sub-population $k$} 
through \Cref{prop:optimal_controls}. The mean-field distributions, however, as discussed in \Cref{def:mean_field_distribution} are not arbitrary and are not to be chosen freely. Specifically, they \textit{must} correspond to the marginal distributions of the agents SREC inventories. As per \eqref{eq:state_dynamics}, the marginal distribution of an agent's SREC inventory is governed by their optimal controls, which  depend on $\bm{\mu}$. 
This circular dependence underlies the core of the MFG methodology, and we must discuss how to find such a $\bm{\mu}$. However, before doing so, we characterize the dependence of the firms optimal controls on the mean-field distribution by finally obtaining $S_t^{\bm{\mu}}$ that clears the market.

\subsection{Equilibrium SREC Price} \label{sec:SREC_price} 
We begin this section by examining the SREC price implied by the mean-field solution for the infinite-player game. Specifically, consider the SREC price that results in the market clearing condition being satisfied, under the assumption that firms behave per \eqref{eq:optG} and \eqref{eq:optGamma}. 

\begin{proposition}[Equilibrium SREC Price in Infinite Player Game] \label{prop:infinite_player_SREC_price}
In the setting of the infinite-player game described in Section \ref{sec:mfg_problem}, we further assume that:
\begin{enumerate}
    \item[(i)] The market clearing condition described in \eqref{eq:infinite_player_clearing_condition} holds, and
    \item[(ii)] Firms behave as per the optimal controls for the mean-field game (i.e. they generate and trade according to \eqref{eq:optG} and \eqref{eq:optGamma}).
\end{enumerate}
The equilibrium 
SREC price is then given by
\begin{equation}
\color{black} S_t^{\bm{\mu}} = P \sum_{k\in\mcK} \eta_k \,\EE\big[Y_t^{\ikrep}\big], \,\text{ a.s.} \label{eq:equilibrium_price}
\end{equation}
where {\color{black}$Y^{(k)}_t$ denotes the non-compliance probability at $t$ for a generic firm in sub-population $k$}, and
\begin{equation}\label{etak}
\eta_k=
\frac{\frac{\pi_{k}}{\gamma^{k}}}{{\displaystyle\sum_{k' \in\mcK}} \frac{\pi_{k'}}{\gamma^{k'}}}\in(0,1),\qquad \forall\,k\in\mcK.
\end{equation}
\end{proposition}

\begin{proof}
We begin with the clearing condition {\eqref{eq:infinite_player_clearing_condition}}. 
 
\begin{align*}
0 &= \lim_{N\rightarrow\infty}\frac{\sum_{i\in\mfN} \Gamma_t^{i, \star}}{N} =  \lim_{N\rightarrow\infty} \sum_{k\in\mcK} \sum_{i \in \mfN_k} \frac{\Gamma^{i, \star}_t}{N} \nonumber \\
&= \lim_{N\rightarrow\infty} \sum_{k\in\mcK} \sum_{i \in \mfN_k} \tfrac{1}{N \gamma^k} \left( P Y_t^i - S_t^{\bm{\mu}}\right)
 = \lim_{N\rightarrow\infty} \sum_{k\in\mcK} \tfrac{1}{N \gamma^k} \sum_{i \in \mfN_k} P Y_t^{i} - S_t^{\bm{\mu}} \lim_{N\rightarrow\infty} \sum_{k\in\mcK} \tfrac{N_k}{N \gamma^k}\,.
\end{align*}
Therefore,
\begin{align}
S_t^{\bm{\mu}} = \frac{\displaystyle\lim_{N\rightarrow\infty}\sum_{k\in\mcK} \frac{P}{N\gamma^k} \sum_{i \in \mfN_k} Y_t^{i}}{\displaystyle\sum_{k\in\mcK} \frac{\pi_k}{\gamma^k}} = \frac{\displaystyle\lim_{N\rightarrow\infty}\sum_{k\in\mcK} \frac{P N_k}{N \gamma^k} \sum_{i \in \mfN_k} \frac{Y_t^{i}}{N_k}}{\displaystyle\sum_{k\in\mcK} \frac{\pi_k}{\gamma^k}} 
= \color{black} P \sum_{k\in\mcK} \eta_k\, \EE\big[Y_t^{\ikrep}\big], \label{eq:infinite_player_SREC_price}
\end{align}
as required. The third equality follows by substituting \eqref{eq:optGamma}, and all subsequent equalities follow from \eqref{eq:proportion}, the definition of $Y^i_t$ \eqref{eqn:defn-Y}, \textcolor{black}{de Finetti's Law of Large Numbers, the lack of common noise,} and limit laws.

\textcolor{black}{For clarity, we point out some notational changes in the last equality, where we replace the limiting mean of $Y_t^\ik$ (for $i\in\mfN_k$) with the expectation of $Y_t^\ikrep$. This notational change is to more concretely emphasize the dependence on the sub-population $k$. It is implicit in the notation $Y_t^\ik$, as $i$ belongs to sub-population $k$, but we make this more explicit when put into the expectation for the reader's benefit.}

\end{proof}

As a reminder, $\color{black} Y_t^{\ikrep}$ may be interpreted as the non-compliance probability of a representative firm in sub-population $k$. As such, \eqref{eq:equilibrium_price} may be interpreted as follows: (i) The expression $\color{black} \EE[Y_t^{\ikrep}]$ is the average  non-compliance probability within sub-population $k$, and (ii) the SREC equilibrium price is a weighted average over sub-populations (scaled by $P$) where the weights are proportional to the ratio of that sub-population and its trading costs $\frac{\pi_k}{\gamma^k}$.
The form of \eqref{eq:equilibrium_price} implies a few notable properties of the SREC price. By \eqref{eqn:defn-Y},  $S_t^{\bm{\mu}} \in [0, P]$ as expected. SREC prices increase as firms increase their non-compliance probabilities, which results from the fact that increasing demand for SRECs (from firms who are far from compliance) increases the price. The opposite holds true as well; as firms decrease their non-compliance probabilities (that is, they are highly likely to comply), the SREC price decreases.

Finally,  while $Y_T^i$ is either zero or one for all $i \in \mathfrak{N}$, depending on whether firm $i$ has sufficient SRECs to comply, it does not imply that $S_t^{\bm{\mu}}$ converges to $0$ or $P$. Instead, the SREC price is a weighted average of non-compliance probabilities across all agents/types.

We now specify the form of the equilibrium price $S_t^{{\mu^{[N]}}}$ in the finite-player game. As previously mentioned, this is derived endogenously, through the optimal behaviours that emerge from the agents optimizing their respective cost functionals, and the market clearing condition \eqref{eq:market_clearing_condition}.

\begin{proposition}[Market Clearing SREC Price in Finite Player Game] \label{prop:finite_player_SREC_price}
In the setting of the finite-player game described in Section \ref{sec:model}, suppose that:
\begin{itemize}
    \item[(i)] The market clearing condition \eqref{eq:market_clearing_condition} holds.
    \item[(ii)] Firms behave as per the optimal controls for the mean-field game (i.e. they generate and trade according to \eqref{eq:optG} and \eqref{eq:optGamma}), modified for the version of the price they observe and $Y^i_t$ in the finite-player game.
\end{itemize}
The market-clearing SREC price is 
\begin{equation}
S_t^{{\mu^{[N]}}} = P\,\frac{\displaystyle\sum_{k\in\mcK} \frac{ N_k}{N\gamma^k} \sum_{i\in\mfN_k} \frac{Y_t^i}{N_k}}{\displaystyle\sum_{k\in\mcK} \frac{N_k}{N \gamma^k}}. \label{eq:finite_player_SREC_price}
\end{equation}
\end{proposition}
\begin{proof}
We note that the market clearing condition in \eqref{eq:market_clearing_condition} is equivalent to the condition $\sum_{i\in\mfN} \Gamma_t^{i, \star} = 0$, as $N$ is finite. We begin with this.
\begin{align}
0 = \sum_{i\in\mfN} \Gamma_t^{i, \star} =  \sum_{k\in\mcK} \sum_{i \in \mfN_k} \Gamma^{i, *}_t 
&=  \sum_{k\in\mcK} \sum_{i \in \mfN_k} \tfrac{1}{\gamma^k} \left( P Y_t^{i} - S_t^{{\mu^{[N]}}}\right) \nonumber \\
& = \sum_{k\in\mcK} \tfrac{1}{\gamma^k} \sum_{i \in \mfN_k} P Y_t^{i} - S_t^{{\mu^{[N]}}} \sum_{k\in\mcK} \tfrac{N_k}{\gamma^k} \nonumber 
\end{align}
Hence,
\begin{align}
S_t^{{\mu^{[N]}}} &= \frac{\sum_{k\in\mcK} \tfrac{P}{\gamma^k} \sum_{i \in \mfN_k} Y_t^{i}}{\sum_{k\in\mcK} \tfrac{N_k}{\gamma^k}}  = \frac{\sum_{k\in\mcK} \tfrac{P N_k}{N \gamma^k} \sum_{i \in \mfN_k} \tfrac{Y_t^{i}}{N_k}}{\sum_{k\in\mcK} \tfrac{N_k}{N\gamma^k}} ,
\end{align}
as required. The second equality arises from substituting \eqref{eq:optGamma}, and each subsequent equality is through simple algebraic operations.
\end{proof}




The equilibrium SREC price in the finite-player game has many of the same properties as the market clearing SREC price in the infinite player game. Namely, it also takes values on the range $[0, P]$, is impacted the same way by the non-compliance probabilities of the regulated firms, and is not necessarily pinned to $0$ or $P$ as $t \to T$. The finite-player market clearing SREC price \eqref{eq:finite_player_SREC_price}, however, is stochastic. 
This is in contrast to \eqref{eq:equilibrium_price}, which is deterministic, as discussed in the remark below.


\begin{remark} \label{rem:deterministic_measure_flow}
It is well known that mean-field games with only minor agents result in a deterministic mean-field measure flow of states (see e.g. \cite{huang2007invariance, carmona2013probabilistic}). Consequently, $\mu_t^{(k)}$ is a deterministic measure flow across $t \in \mathfrak{T}$, for all $k \in \mcK$. As a result, $S_t^{\bm{\mu}}$ is also deterministic. \textcolor{black}{Moreover, $S_t^{\bm{\mu}}$ can be thought of as an expectation (across agents) of a Doob martingale (which is indexed by agents), hence it is in fact constant. When applying the MFG solution to the finite population case, however, the SREC price will inherit stochasticity as discussed in Section \ref{finite_player_game}.}
\end{remark}

\begin{remark}
Prior work in equilibrium pricing for C\&T find that the SREC price is proportional to the probability that the aggregate SREC inventory meets the aggregate requirement (\cite{hitzemann2018equilibrium}). This contrasts with our result, as here we consider trading frictions that prevent an agent from instantaneously changing their inventory with zero cost.
\end{remark}


\subsection{An FBSDE for the Mean-Field Controls} 
\label{sec:FBSDE} 

While we have obtained the form of the agents' optimal controls (given  $\bm{\mu}$), we have not yet solved for them. Additionally, we have not specified the mean-field distributions $\bm{\mu}$. 

Recall that agent-$i$'s optimal controls are governed by \eqref{eq:optG} and \eqref{eq:optGamma}, and are affine transformations of $Y_t^\ik$ and the clearing price. Additionally, recall that, for each $k\in\mcK$, $(\mu_t^{(k)})_{t \in \mfT}$ represents the mean-field distributions of $X_t^\ik$ across agents $i\in\mfN_k$. 

To fully specify the optimal controls of a regulated LSE in an SREC market, we must characterize $Y^\ik_t$ as well as $\bm{\mu}$. Specifically, we need to find the fixed-point of $\bm{\mu}$ such that the optimal controls given the choice of $\bm{\mu}$ imply a controlled state process $X_t^\ik$ over $i$ with  distribution  $\bm{\mu}$.
To achieve this, we first formulate $Y_t^\ik$ as a component of a solution to the FBSDE specified in \Cref{cor:fbsde}.

\begin{corollary} [FBSDE Formulation] \label{cor:fbsde}
The triple $(X_t^{i}, Y_t^{i}, Z_t^{i})$, given $(\bm{\mu}_t)_{t\in\mfT}$, for $i\in\mfN_k$, satisfies the FBSDE 

\begin{subequations}
\begin{align}
dX_t^{i} &= \left(h_t^k + \upsilon^k\,P  Y_t^i
- P\tfrac{1}{\tgamma^k} \sum_{k\in\mcK} \tfrac{\pi_k}{\gamma^k} \,\textcolor{black}{\EE\big[Y_t^{\ikrep}\big]}\right) \,dt + \sigma_t^k dW_t^i, 
\label{eq:HetFBSDE_fwd} 
\\
X_0^{i} &= \xi^i,\label{eq:HetFBSDE_fwd_IC}
\\
dY_t^i &= Z_t^{i} \,dW_t^i,  \label{eq:HetFBSDE_bwd}
\\
Y_T^{i} &=F_{\delta}^\prime(R^k - X_T^{i}) ,
\end{align}
\label{eqn:FBSDE-full}%
\end{subequations}%
where
\begin{equation}
\tgamma^k := \gamma^k \sum_{k'\in\mcK} \frac{\pi_{k'}}{\gamma^{k'}},
\qquad
\upsilon^k := \frac{1}{\gamma^k} + \frac{1}{\zeta^k},
\end{equation}
and $\xi^i\sim \mu^{(k)}_0$.
\end{corollary}
In \eqref{eqn:FBSDE-full}, there are $K$ distinct forms of equations (each agent within a sub-population satisfies the same form of FBSDE).
\begin{proof}
Substitute \eqref{eq:optG}, \eqref{eq:optGamma}, and \eqref{eq:equilibrium_price} into \eqref{eq:state_dynamics} to obtain \eqref{eq:HetFBSDE_fwd}. Observe that $Y_t^i$ is a martingale, and hence can be expressed as in \eqref{eq:HetFBSDE_bwd} (see Chapter 6 of \cite{pham2009continuous}).
\end{proof}


Before proceeding, we make two remarks.

\begin{remark} \label{rem:fbsde}
From \Cref{prop:optimality}, the G\^ateaux derivative vanishes if and only if \eqref{eq:actions} is satisfied, hence, when $\bm{\mu}$ is given, there is an equivalence between the variational problem and the FBSDE \eqref{eqn:FBSDE-full}. Therefore, when $\bm{\mu}$ is given, we are guaranteed a unique solution to the FBSDE \eqref{eqn:FBSDE-full}. However, this does not guarantee that a $\bm{\mu}$ exists such that the unique solution to the FBSDE given said $\bm{\mu}$ results in a controlled state process with marginal distribution $\bm{\mu}$. 
\end{remark}

\begin{remark} \label{rmk:markov_lipschitz_adjoint}
We note that \eqref{eqn:FBSDE-full} falls within the framework in \cite{delarue2002existence}, and the set of assumptions (A1) therein are clearly satisfied. By Corollary 1.5 of \cite{delarue2002existence}, for $\bm{\mu}\in\mcP(\R)$,  $Y_t^i$ has a Markov form, i.e., there exists a function $Y^{(k)}:\R_+\times \R\to \R$ s.t. $Y_t^i = Y^{(k)}(t, X_t^i; \bm{\mu})$. Moreover, $Y^{(k)}(t, X_t^i; \bm{\mu})$ is Lipschitz continuous in its  arguments. While $\bm{\mu}$ is not an argument, we include it to emphasize  that $Y^{(k)}(\cdot,\cdot;\bm\mu)$  depends on the, at this point exogenously, specified choice of $\bm{\mu}$.
\end{remark}
As a consequence of \Cref{rmk:markov_lipschitz_adjoint}, we may re-state the infinite-population market clearing price \eqref{eq:equilibrium_price} as 
\begin{equation}
    S_t^{\bm{\mu}} = P \sum_{k\in\mcK} \eta_k \int Y^{(k)}(t, x; \bm{\mu})\; \mu_t^{(k)}(dx), \label{eq:equilibrium_price_mkv}
\end{equation}
with $\eta_k$ retaining its definition \eqref{etak} from \Cref{sec:SREC_price}.

Returning to the FBSDE \eqref{eqn:FBSDE-full}, we next aim to address the question of the existence (and possibly uniqueness) of a fixed point of $\bm{\mu}$ to it.
To be precise, we seek a mean-field distribution $\bm{\mu}$ and a progressively measurable triple $(X^i, Y^i, Z^i) = (X_t^i, Y_t^i, Z_t^i)_{t \in \mathfrak{T}}$ that satisfies \eqref{eqn:FBSDE-full}, such that $\mu_t^{(k)}$ coincides with $\mathcal{L}(X_t^i)$ for all $i \in \mfN_k$, for all $k \in \mcK$. This is a fixed point problem on the space of measure flows.

At this fixed point, $\mu_t^{(k)}$ is replaced with $\mathcal{L}(X_t^i)$ in \eqref{eqn:FBSDE-full}, and results in the FBSDE becoming of McKean-Vlasov FBSDE (MV-FBSDE) type. Such equations are characterized by their dependence on the law of its solution and arise in the probabilistic formulation of mean-field games. For more details on MV-FBSDEs, see Chapter 3.2.2 of \cite{carmona2018probabilistic}. Finding this fixed point is equivalent to finding a solution to the MV-FBSDEs. Armed with such a solution, we arrive at the full characterization of all agents' optimal actions \eqref{eq:problem_statementMF} and the resulting equilibrium SREC price \eqref{eq:equilibrium_price}. We now investigate whether such a fixed point exists, and if so, whether it is unique.

\subsection{Existence and Uniqueness} \label{sec:existence_uniqueness}

In this section we study existence and uniqueness of a solution to the MV-FBSDE. We first discuss existence. 
\begin{proposition}[Existence of a Fixed Point] \label{prop:existence}
There exists a mean-field distribution $\bm{\mu}$ and a progressively measurable triple $(X^i, Y^i, Z^i) = (X_t^i, Y_t^i, Z_t^i)_{t \in \mathfrak{T}}$ that satisfy \eqref{eqn:FBSDE-full}, such that $\mu_t^{(k)}$ coincides with $\mathcal{L}(X_t^i)$ for all $i \in \mfN_k$, for all $k \in \mcK$.
\end{proposition}
\begin{proof} See \Cref{Proof:prop:existence}.
\end{proof}
Next we consider uniqueness.~
\begin{proposition} [Uniqueness of a Fixed Point] \label{prop:uniqueness}
There is at most one solution to the MV-FBSDE associated with \eqref{eqn:FBSDE-full}.
\end{proposition}
\begin{proof} See \Cref{Proof:prop:uniqueness}
\end{proof}

As a result of establishing the existence and uniqueness of a solution to the MFG, we have established the existence and uniqueness of a Nash equilibrium to the problem.
\begin{theorem} [Nash Equilibrium] \label{thrm:nash_eq}
The set of optimal controls given by \eqref{eq:optG}, \eqref{eq:optGamma}, along with progressively measurable triple $(X^i, Y^i, Z^i) = (X_t^i, Y_t^i, Z_t^i)_{t \in \mfT}$ and the measure flow $\bm{\mu}$ that satisfy the MV-FBSDE version of \eqref{eqn:FBSDE-full}, forms the unique Nash equislibrium for the MFG problem. 
\end{theorem}
\begin{proof}
From \Cref{prop:existence} and \Cref{prop:uniqueness}, we obtain the existence and uniqueness of $\bm{\mu}, X^i, Y^i, Z^i$ that satisfy the MV-FBSDE version of \eqref{eqn:FBSDE-full}. Moreover, $\bm{\mu}$ is consistent with the marginal distribution of $X^i$ where agents behave per \eqref{eq:optG} and \eqref{eq:optGamma}. These behaviours optimize each agent's cost functional for a given $\bm{\mu}$ and are unique, by \Cref{prop:optimal_controls}. The consistency of $\bm{\mu}$ ensures that  agents' controls remain optimal as we flow through time, as the mean field distribution implied by agents' behaviour through time is exactly the measure flow $\bm{\mu}$. This holds for all agents, and therefore, we have found the set of unique controls that satisfy $\eqref{eq:problem_statementMF}$, and thus, the unique Nash equilibrium.
\end{proof}

Having established the existence and uniqueness of a set of controls which result in a Nash equilibrium for the MFG, we next relate this solution to a form of optimality for the finite-player game.
 
\section{The $\epsilon$-Nash Property} \label{sec:epsilon-nash}

Thus far, we have considered the mean field game version, i.e., the infinite-player limit, of our original finite player problem.  This section demonstrates that  the solution of the mean field game is a sensible approximation to the original problem. Specifically, we demonstrate that solving the mean field problem \eqref{eq:infinite_player_clearing_condition}-\eqref{eq:problem_statementMF} corresponds to an $\epsilon$-Nash equilibrium for the finite-player problem \eqref{eq:problem_statement}- \eqref{eq:market_clearing_condition}.

From the previous section,  the solution to \eqref{eq:infinite_player_clearing_condition}-\eqref{eq:problem_statementMF} corresponds to the solution to a fixed point problem involving the FBSDE \eqref{eqn:FBSDE-full}. In \Cref{sec:FBSDE}, we established the existence and uniqueness of the solution to this fixed point, i.e. we demonstrated that there exists a unique mean field distribution $\bm{\mu}$ such that the optimal controls implied by \eqref{eq:optG} and \eqref{eq:optGamma} imply a solution to the FBSDE \eqref{eqn:FBSDE-full} that has a marginal distribution coinciding with $\bm{\mu}$.
The existence of this solution implies that we have found $\{\{g^i,\Gamma^i\} \in \mcA^i\}_{i = 1}^\infty$ that form a Nash equilibrium for the infinite player game, as per \Cref{thrm:nash_eq}. The solution to the mean field problem does not, however, form a Nash equilibrium for the finite-player problem. 
For the finite player game, we assume agents use (mean field) optimal strategies applied to their observable states. That is, we assume their strategies are
\begin{subequations}
\begin{align}
    g_t^{i, \star}  = h_t^k + \tfrac{P}{\zeta^k} Y^{(k)}(t, X_t^{i, [N]}; \bm{\mu}), \quad \text{and} \quad
\Gamma_t^{i, \star} =   \tfrac{1}{\gamma^k} \left( P \, Y^{(k)}(t, X_t^{i, [N]}; \bm{\mu}) - S_t^{{\mu^{[N]}}}\right).
\label{eqn:finite-opt-controls}
\end{align}
Here, the equilibrium price is given by the finite-player clearing condition
\begin{equation}
    S_t^{\mu^{[N]}} = P
    \frac{\displaystyle\sum_{k\in\mcK} \tfrac{ N_k}{N \gamma^k} \sum_{i \in \mfN_k} \tfrac{Y^{(k)}(t, X_t^{i, [N]}; \bm{\mu})}{N_k}}
    {\displaystyle\sum_{k\in\mcK} \tfrac{N_k}{N\gamma^k}}. 
    \label{eqn:S-finite-clearin}
\end{equation}%
\label{eqn:finite-Opt-Controls-and-clearing}%
\end{subequations}%
Recall $X_t^{i, [N]}$ is the agents' finite-population state process solving \eqref{eq:state_dynamics}, with the SREC price above (implicitly assuming all agents are behaving per \eqref{eqn:finite-opt-controls}). We denote $((X^i,Y^i)_{i\in\mfN},\bm{\mu})$ as the solutions to  \eqref{eqn:FBSDE-full} such $\mathcal{L}\left(\{X_t^i\}_{i\in\mfN_k}\right)=\mu^{(k)}$, $\forall\;k\in\mcK$, with $Y^{(k)}$ as the Markov representation of $Y^i$ for agent $i$ in sub-population $k$.

\begin{remark}
The difference between $X_t^{i, [N]}$ and $X_t^i$ is subtle yet important. The former is the forward state described through \eqref{eq:state_dynamics}, which implies that the agent is acting on the finite-player SREC price. Meanwhile, $X_t^i$ is a component of the solution to \eqref{eqn:FBSDE-full}, which implies that the agent is behaving according to \eqref{eq:optG} and \eqref{eq:optGamma}, and in particular, is using the deterministic path of the SREC price \eqref{eq:equilibrium_price} process (or equivalently \eqref{eq:equilibrium_price_mkv}) implied through the infinite-player game.
\end{remark}

\begin{remark}
In this section, we express the (finite-player) performance criterion using the notation $J^{i}(g^i, {\Gamma}^i, \bm{g}^{-i}, \bm{\Gamma}^{-i})$, where $\bm{g}^{-i} = (g^1, ..., g^{i-1}, g^{i+1}, ..., g^{N})$ and $\bm{\Gamma}^{-i}$ is defined analogously.
\end{remark}
\begin{definition} ($\epsilon$-Nash equilibrium)
 In the finite-player ($N$-player) game, a set of controls  $\mathcal{U}^\star=\{(g^{i,\star}, \Gamma^{i,\star})\in \mcA^i,  i \in \mfN\}$ forms an $\epsilon$-Nash equilibrium for the collection of objective functionals $\{J^{i}, i \in \mfN_k, k \in \mcK \}$ if there exists an $\epsilon > 0$
such that 
\begin{equation}
\begin{split}
    J^{i}(g^{i,\star}, \Gamma^{i,\star}, \bm{g}^{-i,\star}, \bm{\Gamma}^{-i,\star}) 
    &\geq \inf_{(g^i, \Gamma^i) \in \mcA^i} J^{i} (g^i, \Gamma^i, \bm{g}^{-i,\star}, \bm{\Gamma}^{-i,\star}) 
    \\
    &\geq J^{i}(g^{i,\star}, \Gamma^{i,\star}, \bm{g}^{-i,\star}, \bm{\Gamma}^{-i,\star}) - \epsilon,
\end{split}
\end{equation}
for all $i \in \mfN_k, k \in \mcK$. 
\end{definition}

This definition says that if an agent unilaterally deviates from the set of strategies $\mathcal{U}^\star$, they can at most benefit by $\epsilon$. We  show that the controls given by \eqref{eq:optG} and \eqref{eq:optGamma}, which are fully specified by the solution to the fixed point problem presented in \Cref{sec:FBSDE} satisfy the $\epsilon$-Nash property once the population size is sufficiently large. The $\epsilon$-Nash property is achieved with respect to the admissible control set $\mcA^{i}$.
Thus, for an $N$-player game, there is a certain tolerance (value of $\epsilon$) that determines how close the decentralized $\mcA^i$-adapted solutions to the MFG problem are to the Nash equilibrium for the finite-player game. We show that $\epsilon$ decreases as $N$ increases.




We must define some additional notation before we proceed. Define
\begin{equation}
\mfN_k^{-i} := \{\, j \in \mfN_k \;| \;j \neq i \,\},
\end{equation}
and $|\mfN_k^{-i}| = N_k^{-i} := N_k - \Id_{i \in \mfN_k}$.
We use the notation $S_t^{\mu^{[N]},-i}$ to denote the clearing price when all agents except for agent-$i$ acts according to \eqref{eqn:finite-opt-controls}, and agent-$i$ uses arbitrary controls $(g,\Gamma)\in\mcA^{i}$. The clearing condition in this case is
\begin{equation}
    \Gamma_t + \sum_{k\in\mcK} \sum_{j \in \mfN_k^{-i}} \Gamma_t^{j, \star}  = 0\,,
\end{equation}
and the equilibrium price is 
\begin{equation}
    S_t^{\mu^{[N]},-i}  
    = P\frac{\displaystyle\sum_{k\in\mcK} \tfrac{ |\mfN_k^{-i}|}{N \gamma^k} \sum_{j\in\mfN_k^{-i}} \tfrac{Y^{(k)}(t, X_t^{j,-i};\bm{\mu}) }{|\mfN_k^{-i}|}}
    {\displaystyle\sum_{k\in\mcK} \tfrac{|\mfN_k^{-i}|}{N\gamma^k}} + \frac{\frac{\Gamma_t}{N}}{\displaystyle\sum_{k\in\mcK} \tfrac{|\mfN_k^{-i}|}{N\gamma^k}}\;.
    \label{eqn:finite-S-without-agent-i}
\end{equation}
Here we again substituted the Markovian form of $Y^j_t=Y^{(k)}(t, X_t^{j,-i}; \bm{\mu})$ from  \Cref{rmk:markov_lipschitz_adjoint}, for all $j\in\mfN$, where the process  $X_t^{j,-i}$ denotes agent-$j$'s state when the SREC  price is the one in  \eqref{eqn:finite-S-without-agent-i}, and agents use the control
\begin{subequations}
\begin{align}
    g_t^{j,-i, \star}  = h_t^k + \tfrac{P}{\zeta^k} Y^{(k)}(t, X_t^{j,-i}; \bm{\mu}), 
    \quad \text{and} \quad
\Gamma_t^{j, -i,\star} =   \frac{1}{\gamma^k} \left( P \, Y^{(k)}(t, X_t^{j,-i}; \bm{\mu}) - S_t^{{\mu^{[N]}},-i}\right).
\end{align}
\end{subequations}

To demonstrate the $\epsilon$-Nash property, we  first establish a useful proposition. 
%
\begin{proposition}
 \label{prop:nash_epsilon_rate}
Take arbitrary admissible controls $(g, \Gamma) \in \mcA^i$  and let $(\bm{g}^{-i,\star}, \bm{\Gamma}^{-i,\star})$ denote the collection of optimal controls given by \eqref{eqn:finite-opt-controls} for agents $j\in \mfN$, $j\ne i$. Suppose that there exists a sequence $\{\delta_N\}_{N=1}^\infty$ such that $\delta_N \rightarrow 0$ and $\left| \tfrac{N_k}{N} -  \pi_k\right| = o(\delta_N)$, for all $k \in \mcK$. Then 
\begin{equation}
    \left\lvert J^{i} (g, \Gamma, \bm{g}^{-i,\star} , \bm{\Gamma}^{-i,\star}) - \overline{J^i}(g, \Gamma; \bm{\mu}) \right\rvert = o\left(\delta_N\right) + o\left(\tfrac{1}{\sqrt{N}}\right)\,, 
    \label{eqn:first-ep-Nash-bound}
\end{equation}
\end{proposition}
\begin{proof}
See \Cref{Proof:prop:nash_epsilon_rate}.
\end{proof}
Now, we may state and prove the $\epsilon$-Nash property for our problem.


\begin{theorem} [$\epsilon$-Nash Property] \label{thrm: epsilon_nash}
Suppose  that (i) \Cref{ass:proportion}-\Cref{ass: MinorContrAction} are enforced, 
(ii) the optimal controls\footnote{We repeat for emphasis and completeness that these optimal controls are only fully specified by the solution to the fixed point problem presented in \Cref{sec:FBSDE}} $\{\bm{g}^\star, \bm{\Gamma}^\star\}$ are as defined in \eqref{eqn:finite-Opt-Controls-and-clearing}; and (iii) there exists a sequence $\{\delta_N\}_{N=1}^\infty$ such that $\delta_N \rightarrow 0$ and $\left| \tfrac{N_k}{N} -  \pi_k\right| = o(\delta_N)$, for all $k \in \mcK$. 
Then
\begin{equation}
\begin{split}
    J^{i}(g^{i, \star}, \Gamma^{i, \star}, \bm{g}^{-i, \star} , \bm{\Gamma}^{-i, \star}) &\geq 
    \inf_{(g^i, \Gamma^i) \in \mcA^i} J^{i} (g^i, \Gamma^i ,\bm{g}^{-i, \star}, \bm{\Gamma}^{-i, \star}) 
    \\
    &\geq 
    J^{i}(g^{i, \star}, \Gamma^{i, \star} ,
    \bm{g}^{-i,\star}, \bm{\Gamma}^{-i,\star}) - \epsilon, 
\end{split}
    \label{eq:nash-epsilon}
\end{equation}
for all $i \in \mfN_k, k \in \mcK$. 
In particular, $\epsilon = o(\tfrac{1}{\sqrt{N}}) + o(\delta_N)$.

\end{theorem}
\begin{proof}
By the definition of the infimum, the left-hand inequality in \eqref{eq:nash-epsilon} is immediately satisfied. It remains to show that the right-hand inequality is also satisfied.

Observe that for any arbitrary controls $(g, \Gamma) \in \mcA^i$, \Cref{prop:nash_epsilon_rate} implies
\begin{equation}
    J^{i} (g, \Gamma, \bm{g}^{-i, \star} , \bm{\Gamma}^{-i,\star}) \geq \overline{J^i}(g, \Gamma) - \varepsilon_N \geq \overline{J^i}(g^\star, \Gamma^\star) - \varepsilon_N,
\end{equation}
where $\varepsilon_N=o(\delta_N) + o(\tfrac{1}{\sqrt{N}})$.
Changing roles of $\overline{J^i}(g^\star, \Gamma^\star)$ and $J^{i} (g^\star, \Gamma^\star, \bm{g}^{-i, \star} , \bm{\Gamma}^{-i,\star})$ in \Cref{prop:nash_epsilon_rate}, we may further write
\begin{equation}
    J^{i} (g, \Gamma, \bm{g}^{-i,\star} , \bm{\Gamma}^{-i,\star}) \geq J^{i} (g^\star, \Gamma^\star, \bm{g}^{-i,\star}, \bm{\Gamma}^{-i, \star}) - 2\,\varepsilon_N.
\end{equation}
As the above holds for arbitrary controls $(g, \Gamma) \in \mcA^i$, we obtain the  inequality
\begin{equation}
    \inf_{(g^i, \Gamma^i) \in \mcA^i} J^{i} (g^i, \Gamma^i ,\bm{g}^{-i,\star}, \bm{\Gamma}^{-i,\star}) \geq J^{i}(g^{i, \star}, \Gamma^{i, \star} ,\bm{g}^{-i,\star}, \bm{\Gamma}^{-i,\star}) - 2\,\varepsilon_N
\end{equation}
and the proof is complete.
\end{proof}

Thus far, we have  established that there is a unique solution to the MFG (in \Cref{sec:existence_uniqueness}), and that the MFG solution provides an $\epsilon$-Nash solution to the finite-player game. As such, all that remains is to solve the McKean-Vlasov version of \eqref{eqn:FBSDE-full}. Solving MV-FBSDEs analytically is difficult, in all but the simplest situations. Rather, we devise a numerical scheme to solve the MV-FBSDE and relegate the presentation and discussion of this scheme to \Cref{appendix:algorithm}.  We  use our numerical scheme  to perform a variety of  experiments and  discuss them next. 

\section{Numerical Experiments} \label{sec:experiments}

In this section, we present the results of numerical experiments which reveal significant qualitative characteristics about how regulated firms behave in this model and the associated implications of this behaviour. We aim for these experiments to be illustrative and provide insight into the agents' optimal behaviour in emissions markets while explicitly accounting for trading frictions and interactions between them, which to the authors' knowledge, has not been done before. 

There are many viable experiments that one could run within this model framework, especially when  considering the potential variety and number of sub-populations of load serving entities that are regulated by SREC markets. 
As such, we restrict focus  to a handful of interesting and meaningful experiments that allow us to draw insights into the drivers of optimal behaviour in SREC markets, the nature of the interaction between agents, and the outcomes of agents acting optimally. 


The scarcity of accessible data introduces a significant challenge for obtaining  calibrated model parameters. 
For example, we do not have data relating to costs firms incur for choosing to generate or trade SRECs. Consequently, we instead focus on highlighting the key features and outcomes  that we believe are applicable to any emissions market system, and aligning said features and outcomes with  economic arguments that justify their validity and applicability. That is, we provide normative rather than informative results.

For the first set of numerical experiments, we use the  parameters reported in Tables \ref{tbl:ComplianceParams} and \ref{tbl:ModelParams}. For simplicity of analysis and interpretation, we choose $h_t^k$ and $\sigma_t^k$ to be constant and set $X_0^i \sim N(\nu_0^k, m_0^k)$ 
for all $i \in \mfN_k$, with the values specified in Table \ref{tbl:ModelParams}. 
\begin{table}[bt]
\centering
\begin{threeparttable}
 \begin{tabular}{cccccc}
 \hline
 $\boldsymbol{n}$ & $\boldsymbol{\Delta t}$ & $\boldsymbol{T}$ & $\boldsymbol{P}$ (\textdollar/SREC) & $\boldsymbol{R^k}$ (SREC) & $\boldsymbol{K}$  \\
 \hline
 1 & $\frac{1}{52}$ &1&  1 & 1 (for $k=1,2$) & 2 \\
\hline
\end{tabular}
\end{threeparttable}
\caption{Compliance Parameters}
\label{tbl:ComplianceParams}
\end{table}

\begin{table}[bt]
\centering
\begin{threeparttable}
\begin{tabular}{cccccccc}
\hline
Sub-population & $\boldsymbol{\pi_k}$ & $\boldsymbol{h^k}$  & $\boldsymbol{\sigma^k}$ & $\boldsymbol{\zeta^k}$ &$\boldsymbol{\gamma^k}$ & $\boldsymbol{\nu_0^k}$ & $\boldsymbol{m_0^k}$\\
\hline
$k = 1$ &  0.25 & 0.2 & 0.1 & 1.75 & 1.25 & 0.6 & 0.1 \\
$k = 2$ &  0.75 & 0.5 & 0.15 & 1.25 & 1.75 & 0.2 & 0.1 \\
\hline
\end{tabular}
\end{threeparttable}
\caption{Model Parameters}
\label{tbl:ModelParams}
\end{table}

As $T = 1, n = 1$, we are focused on a single-period SREC market that takes place over a year. Firms must submit $1$ SREC at $T$, and pay a penalty of $\$ 1$ for every SREC they fail to submit. We assume $K = 2$, with each sub-population having parameters detailed in Table  \ref{tbl:ModelParams}. Specifically, we assume a 3:1 ratio in favour of sub-population $2$. When comparing the two sub-populations, we see that sub-population $2$ has a greater base generation rate ($h^k$) and a slightly higher volatility associated with generation ($\sigma^k$) than sub-population $1$. They also have lower generation costs ($\zeta^k$), while experiencing a higher level of trading costs ($\gamma^k$). Finally, we see the average initial inventory of firms in sub-population $2$ is lower than in sub-population $1$.


\subsection{Optimal Behaviour Summary}

\begin{figure}[bt]
\centering
\includegraphics[align = c, width=0.75\textwidth]{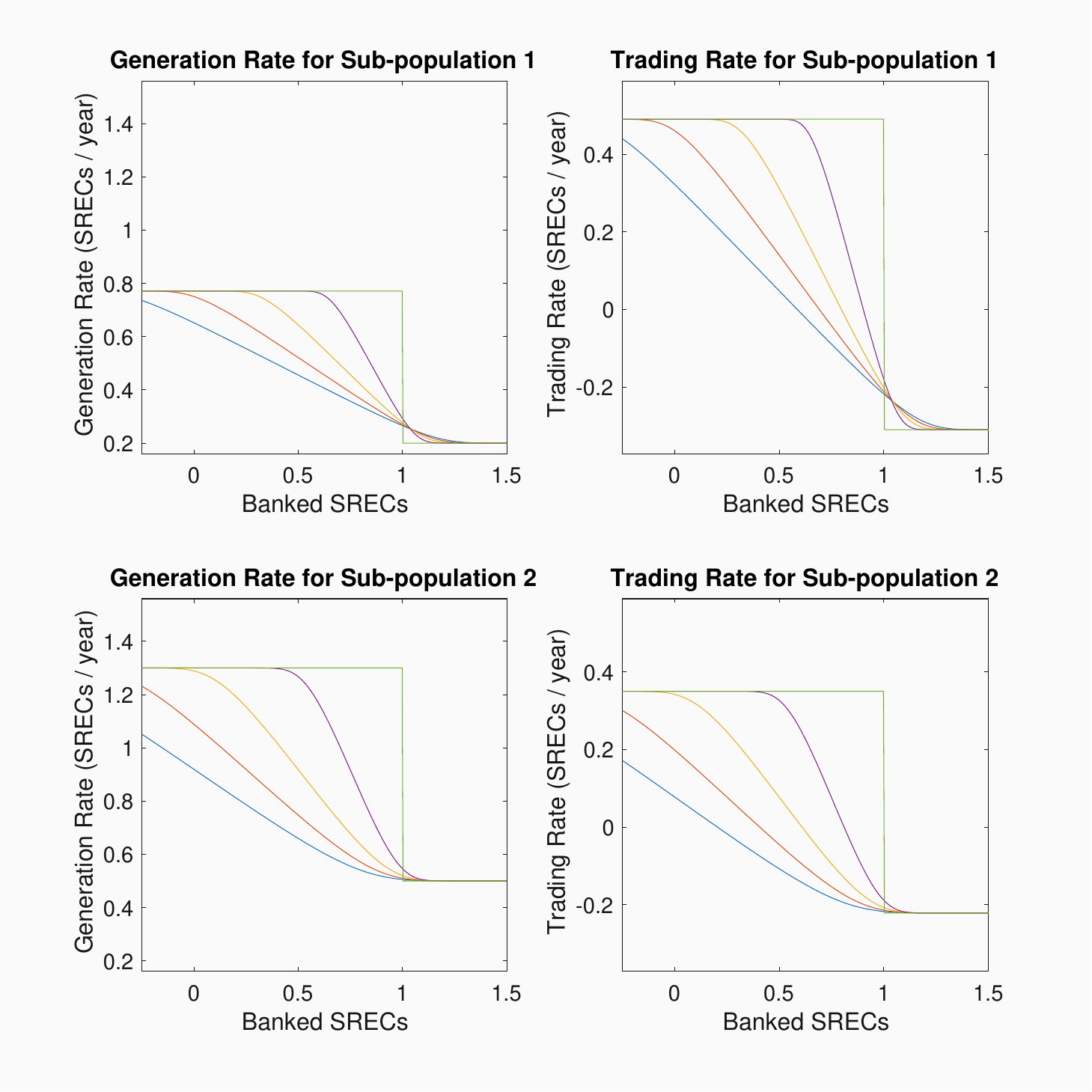}
\includegraphics[align = c, width=0.1\textwidth]{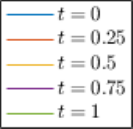}
\caption{Optimal firm behaviour as a function of banked SRECs for various time-steps, for a representative firm from each sub-population. Parameters in Tables \ref{tbl:ComplianceParams} and \ref{tbl:ModelParams}.}
\label{fig:opt_beh}
\end{figure}
A regulated firm's optimal behaviour is one of the key outputs from \eqref{eqn:FBSDE-full}. \Cref{fig:opt_beh} shows the dependence of the optimal trading and generation rate on banked SRECs for a representative agent of each sub-population through time. The equilibrium SREC price $S_t^{\bm{\mu}}$ is implicitly accounted for through the solution of the MV-FBSDE \eqref{eqn:FBSDE-full}.
In this mean-field limit, individual firms have no impact on the price. 

One clear pattern seen in \Cref{fig:opt_beh}  is  the existence of distinct regimes of generation / trading. This  pattern  also arises in the single-agent problem (see \cite{shrivats2019behaving}). At low levels of SREC inventory, and towards the end of the period, firms in both sub-populations generate / purchase until the marginal cost of producing / purchasing an additional SREC exceeds $P$. This occurs because, in this case, the firm almost certainly fails to comply. By generating or purchasing an additional SREC, the firm avoids paying an additional penalty of $P$, and as such, it is in their best interests to do so until the cost of generation / trading exceeds that value. This follows the classic microeconomic principle of conducting an activity until the marginal benefit of that activity equals its marginal cost. 

As a firm obtains more SRECs, they reach a point where the marginal benefit from an additional SREC is less than $P$. Specifically, as the probability of compliance increases from near $0$, additional SRECs no longer provide a marginal benefit of $P$ (as additional SRECs above $R^k$ must either be sold for less than $P$, or expire worthless). As such, it is no longer necessary for firms to generate / purchase in such quantities, leading to a decrease in both quantities. Once again, the firm adjusts its behaviour so that the marginal costs of purchasing and trading are in line with the marginal benefit the firm receives from doing so. This may also lead to the firm selling SRECs as opposed to purchasing them - in this case, the net proceeds from sales exceed the value of an additional SREC to the firm. 

This decrease occurs until the firm reaches a point where it no longer benefits from additional SRECs. Specifically, this means that an additional SREC does not impact a firms compliance probability (implying their compliance probability is $1$), nor can they sell it for profit. Consequently, we observe another plateau where the firm plans to generate at its baseline $h^k$ and plans to sell SRECs at the value for which the marginal revenue from the sale equals the cost of selling. The  previous observations are consistent for each sub-population.

The optimal firm behaviour detailed above, along with their initial inventory (determined by the distribution of $X_0^i, i \in \mfN_k$) also imply a mean-field distribution of SREC inventory levels across agents of each sub-population. We examine this distribution across each agent sub-population and for various times $t \in \mathfrak{T}$ in \Cref{fig:mean_field_distributions}.
\begin{figure}[bt]
\centering
\includegraphics[align = c, width=0.7\textwidth]{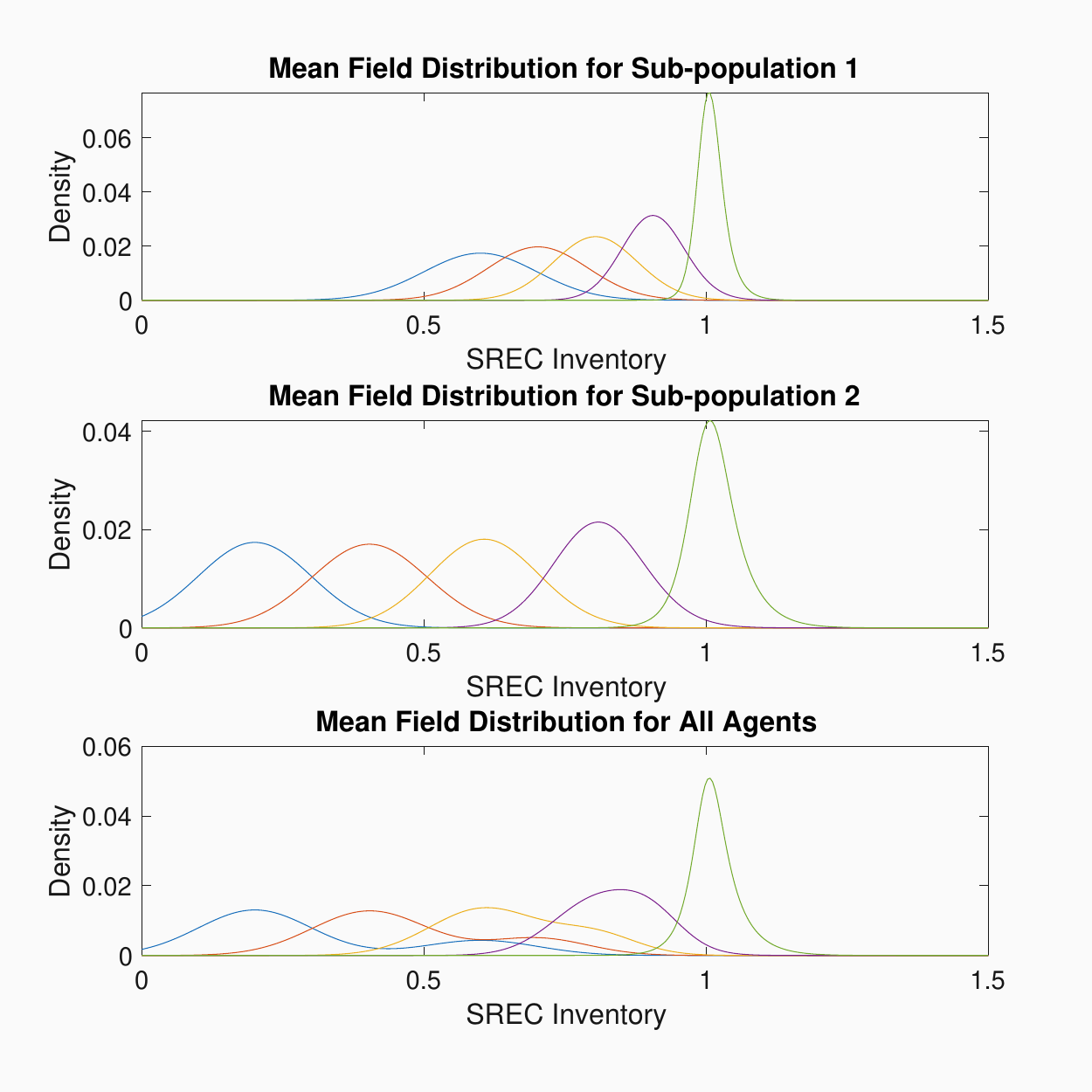}
\includegraphics[align = c, width=0.1\textwidth]{legend_mfg.pdf}
\caption{Distribution of SREC inventory over time for each agent sub-population (sub-population 1 in top row, sub-population 2 in middle row, all agents in bottom row).  Parameters in Tables \ref{tbl:ComplianceParams} and \ref{tbl:ModelParams}.}
\label{fig:mean_field_distributions}
\end{figure}

From \Cref{fig:mean_field_distributions}, we can make some notable observations. For both sub-populations, we see the initial distribution is concentrated around their respective mean $\nu_0^k$. As time progresses, mass shifts upwards, as firms accumulate SRECs, culminating with the mass being concentrated around $R^k$ as $t$ approaches $T$. This is also reflected in the bottom panel, which shows the distribution of SREC inventory across all agent types, incorporating the relative frequency of each sub-population. It shows the bimodal nature of the initial distribution  slowly converge to a unimodal distribution by the time the compliance period ends, with the mode occurring near the requirement.

The patterns detailed above are expected, as they indicate  agents aim to satisfy the compliance requirements. Specifically, about $62\%$ of agents in sub-population $1$ and $61\%$ of agents in sub-population $2$ meet the requirement. It is worth noting that the optimal behaviour for firms does not result in every firm complying exactly.

The primary cause of this phenomenon is that firms do not have absolute control over their SREC production. While the firm plans to generate at a rate of $g^i$, they over or under-generate due to, e.g., variation in sunlight. If a firm severely under-generates, they may not be able to comply. The market frictions imposed by our model, such as trading costs also constrain firms' behaviours, which may also force non-compliance.

Additionally, in a single-period SREC model with no banking, a terminal SREC inventory above the requirement necessarily implies some SRECs are squandered,
as unused SRECs at time $T$ are worthless. As such, the funds used to acquire these spare SRECs (either through generation or trading) are also wasted. This means that it is typically not optimal to accumulate far more than $R^k$ SRECs, as  firms must strike a balance between securing compliance and risking unnecessary investment. The natural incentive for all firms is to just barely satisfy the requirement, as opposed to exactly satisfying it (which would waste no SRECs), so that they are partially protected from a potential random under-generation.



As alluded to earlier, in the mean-field limit, the equilibrium SREC price that the firms transact at is not impacted by any individual agent. Rather, as per \eqref{eq:equilibrium_price}, it is determined by the collective distribution of agents. Consequently, the equilibrium SREC price $S_t^{\bm{\mu}}$ is  an output of our algorithm, as opposed to a state variable or an input. The initial distribution of agent states amongst sub-populations and the optimal behaviour of each regulated agent throughout the period implies a value for $S_t^{\bm{\mu}}$ over time. With the parameter choices made in Tables \Cref{tbl:ComplianceParams} and \Cref{tbl:ModelParams}, the $S_t^{\bm{\mu}}$ implied by the mean-field optimal controls is nearly constant throughout the compliance period, at a value of roughly 0.39. 


Careful examination of the form of \eqref{eq:equilibrium_price} can provide insight as to why this occurs. First, from general results (see \cite{huang2007invariance}), we know that the mean-field measure flow is deterministic. Second, the only term which potentially varies over time in the formula for $S_t^{\bm{\mu}}$ is $\int Y^{(k)}(t, x; \bm{\mu}) \mu_t^{(k)}(dx)$ and represents the expected non-compliance probability of an agent in sub-population $k$. The near-constant nature of $S_t^{\bm{\mu}}$ indicates that this quantity does not change much over time in the infinite-player limit. This suggests that, in the infinite-player limit, the mean-field distribution $\mu_t^{(k)}$ changes over time in such a way that such that expected non-compliance probabilities are essentially constant. That is, the trajectory of the average non-compliance probability across each sub-population of agents at each time $t$ is constant. When we explore the finite-player game in \Cref{finite_player_game}, we show that  the empirical distribution of agents' states is stochastic, in contrast to the mean-field limit. 

\subsection{Finite-Player Simulation} \label{finite_player_game}

While the previous subsections focus on studying the behaviour and properties of the solution to the infinite-player stochastic game, we now turn our attention to the finite-player game. We study the infinite player game because it is more tractable than the finite player stochastic game. However, it is worth numerically studying how the infinite-population solution fares when applied to a more realistic finite-population game, and the associated implications on the SREC price (which is now the market clearing SREC price described by \eqref{eq:finite_player_SREC_price}, as opposed to the equilibrium SREC price described by \eqref{eq:equilibrium_price}). 

We simulate a compliance period for $N$ firms, using the same model and compliance parameters as in \Cref{tbl:ComplianceParams} and \Cref{tbl:ModelParams}, and letting $N = 2000$. This is not a particularly large number of firms - the Sweden-Norway REC market regulated upwards of 9,000 solar firms in 2018 (see \cite{sweden_electricity}). That is, we assume $500$ firms from sub-population 1 and $1500$ firms from sub-population 2. For each firm, we draw their initial SREC inventory from $\mu_0^k$ (using the appropriate sub-population), and simulate forward. At each time-step, each firm chooses their planned generation and trading behaviour, as per \eqref{eq:optG} and \eqref{eq:optGamma}, while being able to observe the SREC price. In doing so, each firm is exposed to their idiosyncratic SREC production noise, which in turn impacts the agents' state in the next time-step, as well as the SREC price. This continues across the entire period. 

We first summarize the overall behaviours of agents of each sub-population by plotting, in \Cref{fig:finite_player_histograms}, histograms of their initial and terminal SRECs, as well as their total generation and trading.
\begin{figure}[h]
\centering
\includegraphics[align = c, width=0.95\textwidth]{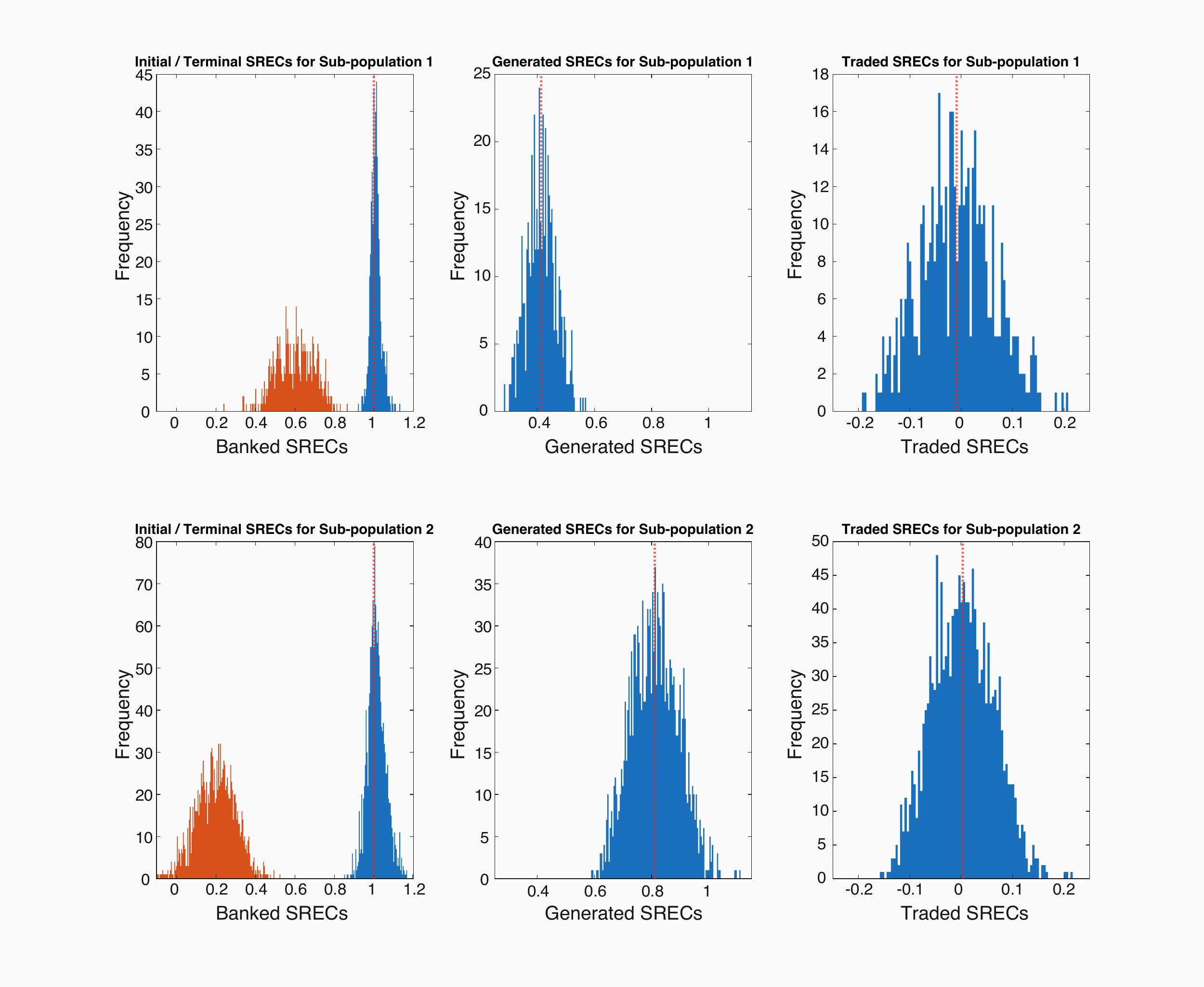}
\caption{Histograms of initial (red) and terminal (blue) SRECs (left panel), generation rate (middle panel), trading rate (right panel). Parameters in Tables \ref{tbl:ComplianceParams} and \ref{tbl:ModelParams}.}
\label{fig:finite_player_histograms}
\end{figure}

As expected, we see that most firms from both sub-populations take actions such that their terminal SREC inventory is near $R^k = 1$. The non-compliance rates for sub-populations $1$ and $2$ are about 0.42 and 0.40, respectively. The means of the total generation and trading for firms from each sub-population are slightly higher than 0.4 and 0.8, respectively. These figures correspond to the average amount of SRECs below the requirement firms from each sub-population begin the compliance period with. We also see that total generation and trading in the finite player game is relatively symmetric.

We next examine  the firms' trading and generation behaviour across the compliance period. To this end, we plot each firm's generation and trading rates over the course of the compliance period, colour coded by the sub-population they belong to, in \Cref{fig:finite_player_controls_sim}.
\begin{figure}[h]
\centering
\includegraphics[align = c, width=0.45\textwidth]{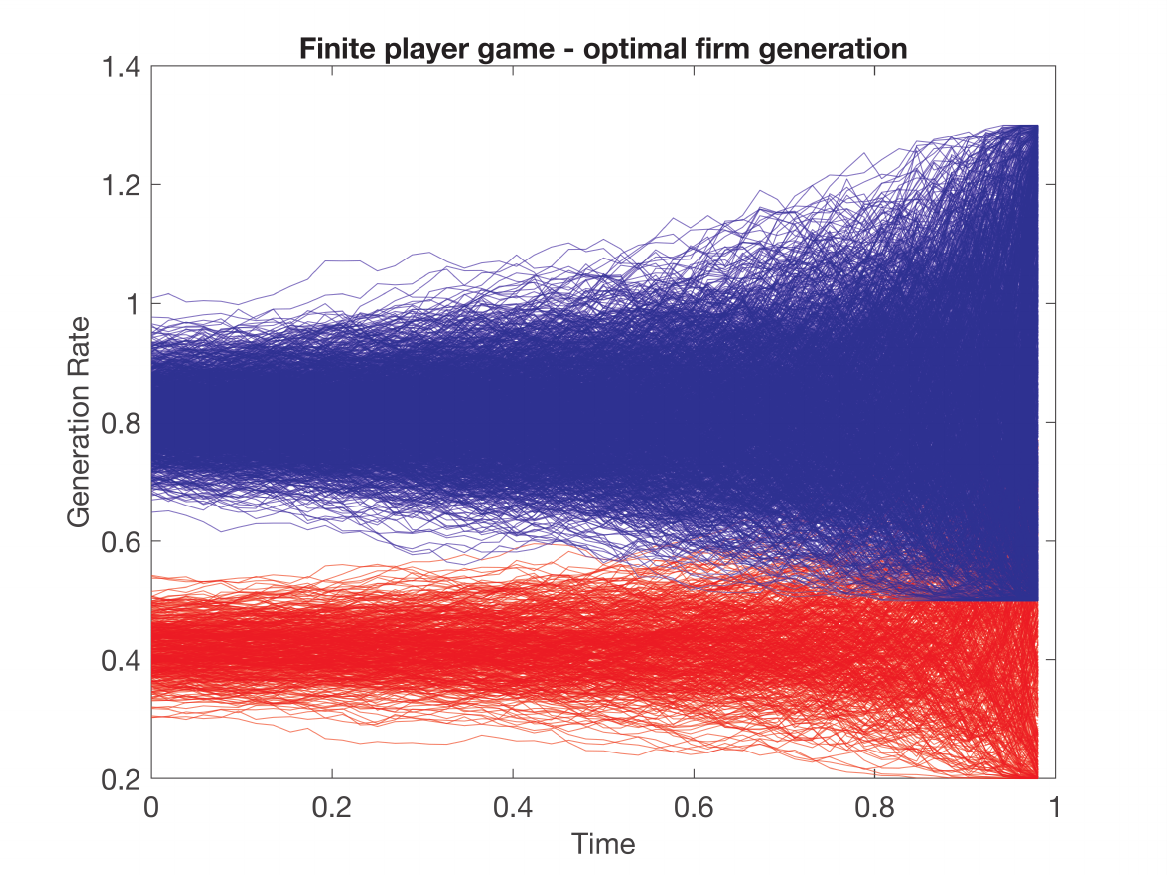}
\includegraphics[align = c, width=0.45\textwidth]{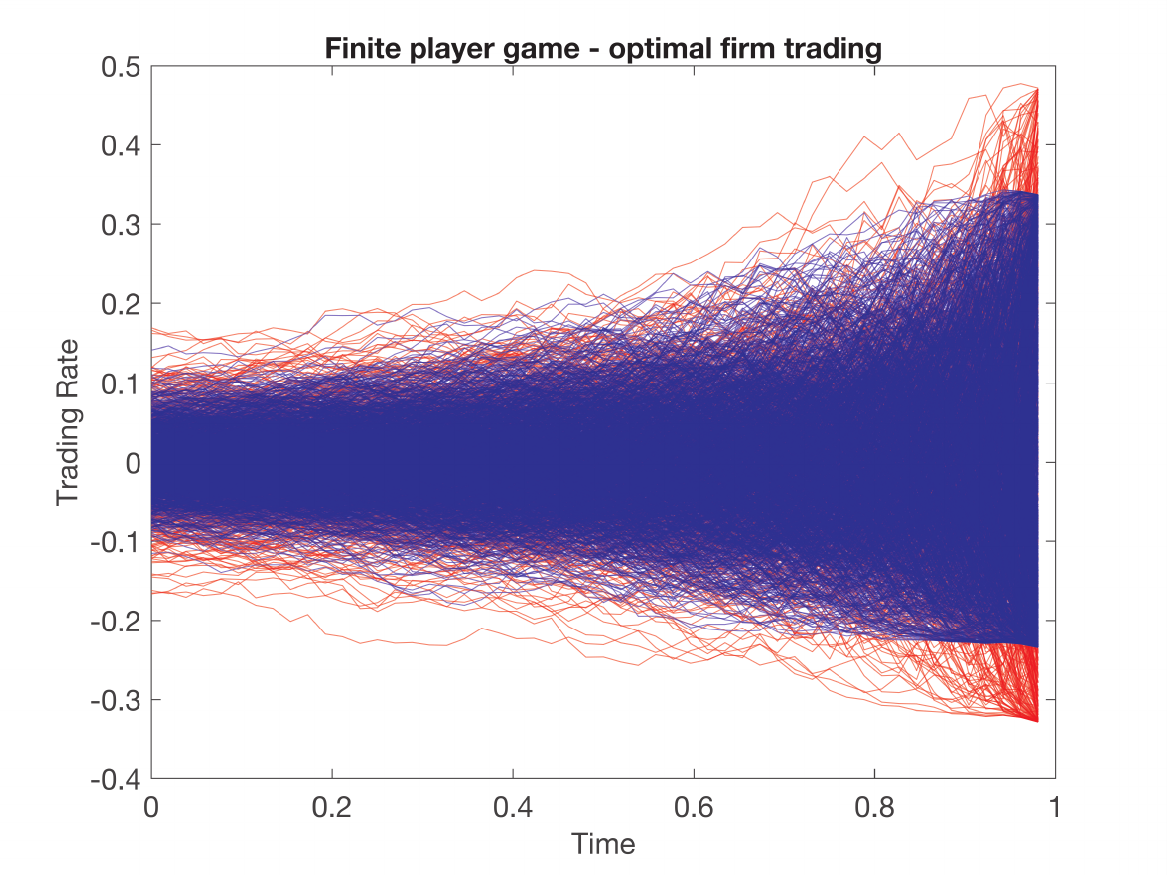}
\caption{Planned generation (left) and trading rates (right) in finite-player simulation with parameters as in \ref{tbl:ComplianceParams} and \ref{tbl:ModelParams}, for sub-population 1 (red) and sub-population 2 (blue).}
\label{fig:finite_player_controls_sim}
\end{figure}

\Cref{fig:finite_player_controls_sim} shows  the differences between the two sub-populations planned generation and trading rates.
As expected, firms from sub-population 1 generate less than those from sub-population 2, due to their lower baseline generation rate and higher cost of generation. The trading rates across firms from different sub-populations largely overlap one another, though the firms from sub-population 1 typically have more extreme behaviour, reflecting their increased capability to participate in the market when compared to firms from sub-population 2.

\begin{figure}[!t]
\centering
\includegraphics[align = c, width=0.7\textwidth]{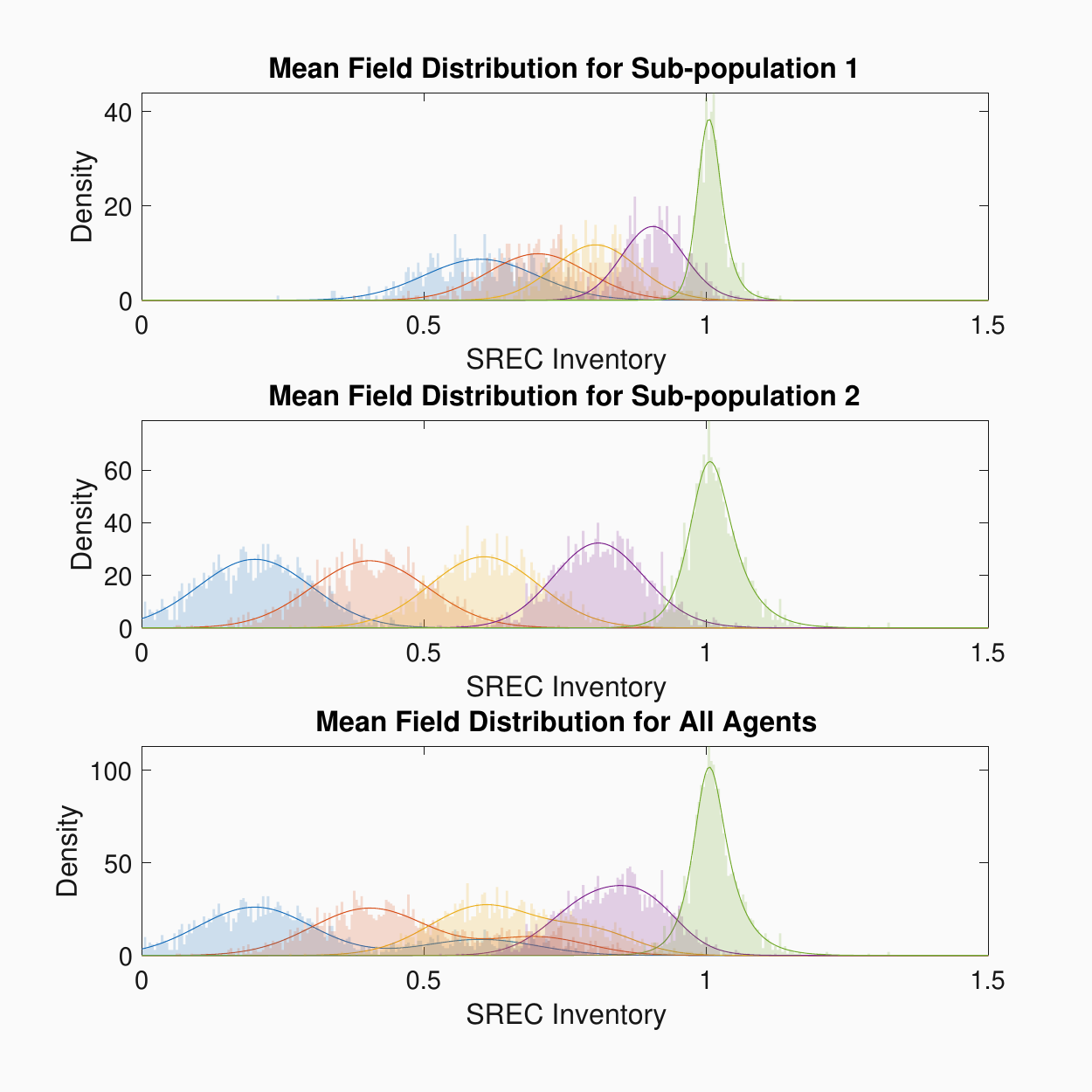}
\includegraphics[align = c, width=0.1\textwidth]{legend_mfg.pdf}
\caption{Empirical and theoretical distribution of SREC inventory over time for each sub-population separately and jointly.  Parameters in Tables \ref{tbl:ComplianceParams} and \ref{tbl:ModelParams}.}
\label{fig:finite_player_distribution}
\end{figure}
Next, we look at the distribution of states across agents to ensure  the distribution of agents in the finite-player sample agrees with those from the mean-field game (see \Cref{fig:mean_field_distributions}). To this end, \Cref{fig:finite_player_distribution} shows histograms of the agents states across the compliance period overlaid on the mean-field distribution (scaled for finite population).

From the figure, we can see that the empirical and theoretical distributions are in alignment with one another, suggesting numerically that the infinite-player solution is a reasonable approximation for the optimal controls of the agents in the finite-player problem, provided the number of agents is sufficiently large. 

\begin{figure}
\centering
\includegraphics[align = c, width=0.48\textwidth]{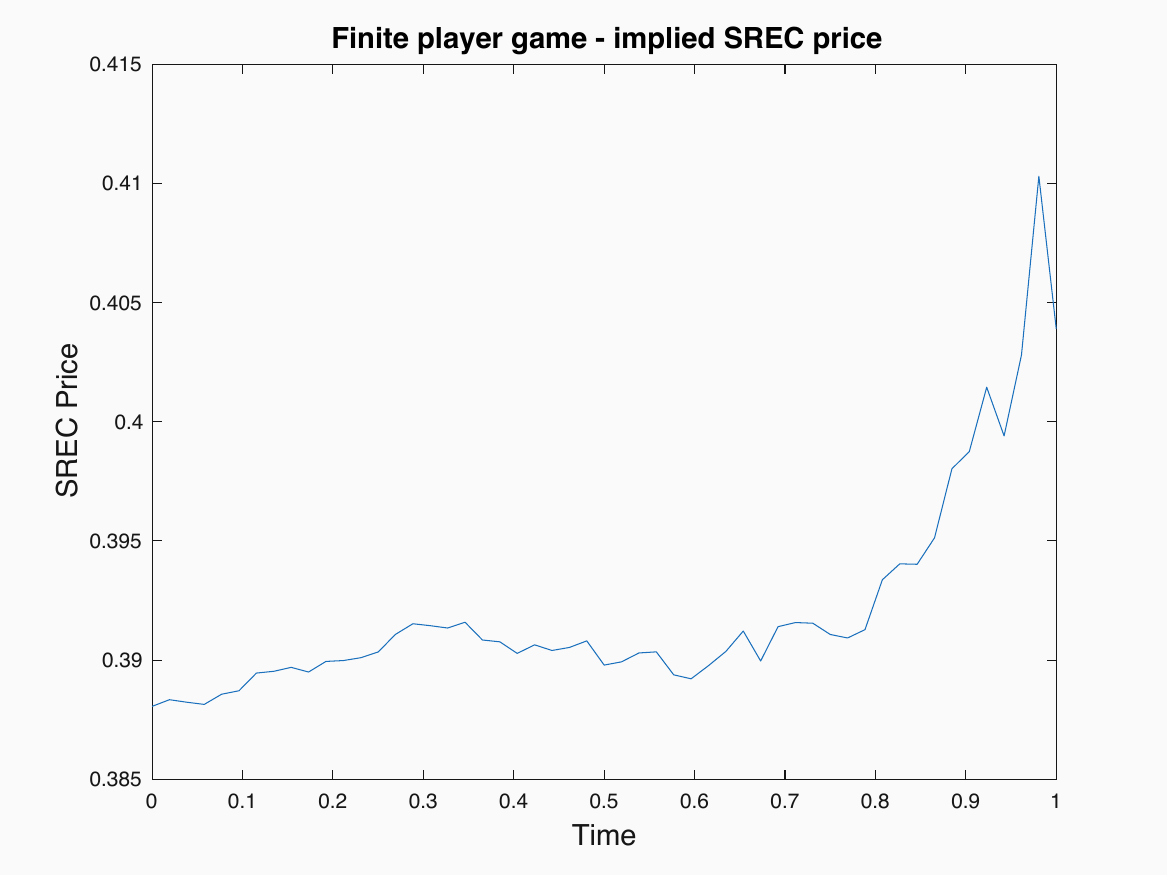}
\caption{Equilibrium SREC price in finite-player setting. Parameters in Tables \ref{tbl:ComplianceParams} and \ref{tbl:ModelParams}.}
\label{fig:finite_player_equilibrium_price}
\end{figure}
Finally, we examine the how the equilibrium SREC price evolves in this finite-player game, when compared to the infinite-player scenario. To do so, we plot the implied equilibrium SREC price, $S_t$, which is endogenously derived at all times $t$ through \eqref{eq:finite_player_SREC_price}, in \Cref{fig:finite_player_equilibrium_price}.

The most immediate difference is that the SREC price now has notable variation. More specifically, this variation is determined by the states of the agents, which themselves depend on the degree of over/under-generation that they experience. 

In this case, we see that the SREC price stays relatively constant for the first three quarters of the compliance period, then rapidly increases over the remainder. This suggests that firms non-compliance probabilities increased over the final parts of the period, resulting in increased demand for SRECs, as firms scramble to ensure they comply. Additionally, we see that the SREC price does not converge to either $0$ or $P$ at time $T$, as is implied in previous work in C\&T as well as SREC markets (see \cite{coulon_khazaei_powell_2015}, \cite{hitzemann2018equilibrium}, \cite{seifert_uhrig-homburg_wagner_2008}, \cite{carmona_fehr_hinz_2009}, and \cite{carmona2010market}). This property is consistent with the comments made in Section \ref{sec:SREC_price} regarding the form of the finite-player market clearing SREC price. 

\section{Conclusion}

We introduce  an $N$-player stochastic game for an SREC market with heterogeneous sub-populations, where each player aims to navigate the market at minimum cost by modulating their planned SREC generation and trading activities. By taking the infinite-player limit of the model, we obtain a MFG for the market, which is a more tractable problem. Using tools from variational analysis, we are able to express the optimal controls of each agent as the solution to a McKean-Vlasov FBSDE. We prove existence and uniqueness of a solution to this FBSDE, and thus we are able to characterize the Nash equilibrium of the MFG. Further, we establish the $\epsilon$-Nash property of the mean field solution when applied to the original finite-player game. Finally, we are able to numerically solve the MV-FBSDE, allowing us to run numerous experiments to better understand the structure of the optimal controls of the agents, and the associated implications on the SREC price and on non-compliance probabilities of each sub-population of agents. 

In particular, we see that the optimal behaviours of agents expressed as a function of their SREC inventory exist in regimes corresponding to the benefit a firm receives by acquiring an additional SREC. This is reinforced when observing the functional forms that the optimal controls take. Additionally, we observe that while the infinite-player SREC price process is deterministic, in the $N$-player stochastic game with agents behaving as per the MFG solution, this is not the case. Instead, we see notable price volatility, particularly towards the end of the compliance period, as firms have less time to react to any unforeseen shocks that impact their compliance probability. 

There are natural areas for improvement, and possible extensions to this work. Our framework and proof methods  can be easily generalized to other MFG with equilibrium pricing settings. In particular, it is possible to generalize the cost functional to be convex but not quadratic. As well, calibration to real-world data remains challenging due to a lack of transparency into the underlying cost functions of regulated LSEs, but is nonetheless a critical next step. Similarly, extending this work to a multi-period SREC market is another area that must be explored further. Nonetheless, in providing the mathematical framework contained in this paper, we have produced a logically consistent and coherent structure under which SREC markets (and indeed, any emissions market) can be studied holistically. We feel such a framework is potentially of great use to regulatory bodies and regulated firms in these systems alike, and aim to continue to refine the work in this area further.


\bibliographystyle{siam}
\bibliography{references}

\appendix

\section{Proofs}\label{appendix:proofs}
\subsection{Proof of \Cref{convexity}} \label{Proof:prop:convexity}
We must show that $\overline{J^i}$ has the following property: for $\lambda \in (0, 1)$ and two sets of controls $(g^i, \Gamma^i), (w^i, \upsilon^i) \in \mcA^i$, we require that
\begin{equation}
\overline{J^i}(\lambda g^i + (1 - \lambda) w^i , \lambda \Gamma^i + (1 - \lambda) \upsilon^i; \bm{\mu}) < \lambda \overline{J^i} (g^i, \Gamma^i; \bm{\mu}) + (1 - \lambda)\overline{J^i} (w^i, \upsilon^i; \bm{\mu}).
\label{eq:convexity_condition}
\end{equation}
To do so, first  write $\overline{J^i}(g^i, \Gamma^i; \bm{\mu})$ 
$= G^i(g^i, \Gamma^i ; \bm{\mu}) + H^i(g^i, \Gamma^i; \bm{\mu})$, 
$i\in\mfN_k$, 
where
\begin{align}
G^i(g^i, \Gamma^i; \bm{\mu}) &:=  \EE\left[\int_0^T
\left\{\begin{pmatrix}g_t^i & \Gamma_t^i \end{pmatrix} Q^k \begin{pmatrix}
g_t^i \\ \Gamma_t^i \end{pmatrix}+ \begin{pmatrix}g_t^i & \Gamma_t^i \end{pmatrix} \begin{pmatrix}
- \zeta^k h_t^k  \\
S_t^{\bm{\mu}}
\end{pmatrix} + \tfrac{\zeta^k}{2} (h_t^k)^2 \right\}dt\right],  
\\
H^i(g^i, \Gamma^i; \bm{\mu}) &:= \EE\left[P \;F_{\delta}\left(R^k - X_0^i - \int_0^T
(g_t^i + \Gamma_t^i) \,dt - \sigma^k W_T^i\right)\right].
\end{align}
We prove that $G^i$ is strictly convex and $H^i$ is convex, which implies that $\overline{J^i}$ is strictly convex.

First consider $G^i$. As $Q_k$ is positive definite, the first term is strictly convex in the controls. The second term is linear in the controls, and thus is also convex. The third term is independent of the controls. Through the linearity of integrals, and since the sum of a strictly convex and a convex function is strictly convex, $G^i$ a strictly convex functional of the controls.

It remains to show that $H^{i}(g^i, \Gamma_i; \bm{\mu})$ is a convex functional of the controls.
\begin{align}
H^i(\lambda g^i + &(1 - \lambda) w^i ,  \lambda \Gamma^i + (1 - \lambda) \upsilon^i; \bm{\mu}) \nonumber \\
\mbox{\tiny{(Definition)}}&= \EE\left[ PF_\delta\left(\underset{Z^i}{\underbrace{R^k - X_0^i - \sigma^k W_T^i}} - \lambda \int_0^T (g_t^i + \Gamma_t^i) dt - (1 - \lambda) \int_0^T (w^i_t + \upsilon^i_t) dt\right)\right]  \nonumber \\
&= \EE\left[PF_\delta\left(\lambda \Big(Z^i - \int_0^T (g_t^i + \Gamma_t^i) dt\Big) + (1 - \lambda)\Big(Z^i - \int_0^T (w^i_t + \upsilon^i_t) dt\Big)\right)\right] \nonumber \\
\mbox{\tiny{($F_\delta(x)$ convex)}}&\leq \EE\left[\lambda PF_\delta\left(Z^i - \int_0^T (g_t^i + \Gamma_t^i) dt\right) + (1 - \lambda)PF_\delta\left(Z^i - \int_0^T (w^i_t + \upsilon^i_t) dt\right)\right]  \nonumber \\
&= \lambda H^i(g^i, \Gamma^i; \bm{\mu}) + (1 - \lambda) H^i(w^i, \upsilon^i; \bm{\mu}) \nonumber
\end{align}
Therefore, $H^i$ is convex, and $\overline{J^i}$ is strictly convex.

\subsection{Proof of \Cref{gateaux_deriv}}\label{Proof:prop:gateaux_deriv}
The G\^ateaux derivative we seek are given by the limits
\begin{subequations}
\begin{align}
\langle \mcD \overline{J^i}(g^i, \Gamma^i; \bm{\mu}), \omega^g \rangle &:= \lim_{\epsilon \rightarrow 0} \frac{\overline{J^i}(g^i + \epsilon \omega^g, \Gamma^i;\bm{\mu}) - \overline{J^i}(g^i, \Gamma^i; \bm{\mu})}{\epsilon} \label{eq:GD_gen_lim},
\quad \text{and}
\\
\langle \mcD \overline{J^i}(g^i, \Gamma^i; \bm{\mu}), \omega^\Gamma \rangle &:= \lim_{\epsilon \rightarrow 0} \frac{\overline{J^i}(g^i, \Gamma^i + \epsilon \omega^\Gamma;\bm{\mu}) - \overline{J^i}(g^i, \Gamma^i; \bm{\mu})}{\epsilon}. \label{eq:GD_trade_lim}
\end{align}
\end{subequations}
We examine each of these limits in turn, starting with \eqref{eq:GD_gen_lim}. To this end,
\begin{align}
\overline{J^i}(g^i + \epsilon \omega^g, \Gamma^i;\bm{\mu}) - \overline{J^i}(g^i, \Gamma^i;\bm{\mu}) = \EE\biggl[&\int_0^T \left\{\tfrac{\zeta^k}{2}(g_t^i - h_t^k + \epsilon\omega_t^g)^2 - \tfrac{\zeta^k}{2} (g_t^i - h_t^k)^2\right\} dt \nonumber 
\\&
+ P\, F_\delta\left(R^k - X_T^i - \epsilon {\textstyle\int_0^T} \omega_t^g dt\right) - P \, F_\delta(R^k - X_T^i) \biggr] \label{eq:GD_gen_numerator}
\end{align}
which follows  from the definition of the cost functional \eqref{eq:agentCostFnRegularizedMF} and the state dynamics of $X_t^i$ \eqref{eq:state_dynamics}. From the Taylor expansion of $F_\delta$, we have
\begin{align}
F_\delta\left(R^k - X_T^i - \epsilon {\textstyle\int_0^T} \omega_t^g dt\right) = F_\delta(R^k - X_T^i)  - \epsilon \int_0^T \omega_t^g \, F_\delta^\prime (R^k - X_T^i) dt + O(\epsilon^2) \label{eq:taylor_exp_terminal}
\end{align}
This holds almost everywhere by the assumptions imposed on on $F_\delta$ (specifically, that its second derivative exists and is bounded a.e.). Here,  $F_\delta^\prime(\cdot)$ denotes the derivative with respect to its argument. Substituting \eqref{eq:taylor_exp_terminal} into \eqref{eq:GD_gen_numerator} and performing some algebra, we obtain
\begin{equation}
\begin{split}
&\overline{J^i}(g^i + \epsilon \omega^g, \Gamma^i;\bm{\mu}) - \overline{J^i}(g^i, \Gamma^i;\bm{\mu}) 
\\ 
&\qquad= \EE\biggl[ \epsilon \int_0^T \zeta^k (g_t^i - h_t^k)\omega_t^g dt - 
 \epsilon \int_0^T \omega_t^g  P\, F_\delta^\prime (R^k - X_T^i)dt + O(\epsilon^2)\biggr]\,. 
\end{split}
\end{equation}
Dividing by $\epsilon$ and taking limits on both sides as $\epsilon \rightarrow 0$ implies
\begin{align}
\langle \mcD \overline{J^i}, \omega^g \rangle = \EE\biggl[\int_0^T \zeta^k (g_t^i - h_t^k)\omega_t^g dt - \int_0^T \omega_t^g  P \, F_\delta^\prime (R^k - X_T^i)dt\biggr]\,.
\end{align}
Next, applying iterated expectations we obtain
\begin{align}
\langle \mcD \overline{J^i}, \omega^g \rangle &= \EE\biggl[\int_0^T \zeta^k (g_t^i - h_t^k)\omega_t^g dt - \int_0^T \omega_t^g  P \, \EE[ F^\prime (R^k - X_T^i) | \Gg_t^i]dt\biggr] \\
&= \EE\left[ \int_0^T \omega_t^g \left( \zeta^k (g_t^i - h_t^k) - P\, Y_t^i\right)dt\right],
\end{align}
as required.
Through the same techniques applied to \eqref{eq:GD_trade_lim}, we obtain \eqref{eq:GD_trade}.

\subsection{Proof of \Cref{prop:optimality}}\label{Proof:prop:optimality}

\noindent\textit{Sufficiency:} Assume \eqref{eq:gen_nec_cond}-\eqref{eq:trade_nec_cond} holds. From \eqref{gateaux_deriv} we have
\begin{equation}
    \langle \mcD \overline{J^i}(g^i, \Gamma^i; \bm{\mu}), \omega \rangle = 0, \quad \forall\, \omega \in \mcA^i,
\end{equation} as required. 

\noindent \textit{Necessity:} 
Assume $\langle \mcD \overline{J^i}(g^i, \Gamma^i; \bm{\mu}), \omega \rangle = 0, \forall \omega \in \mcA^i$. We next show that \eqref{eq:gen_nec_cond}-\eqref{eq:trade_nec_cond} hold. We prove this via the contrapositive. 

Assume that \eqref{eq:gen_nec_cond}-\eqref{eq:trade_nec_cond} do not hold. Then 
\begin{equation}
    B = \left\{ (w, t) \in \Omega \times \mathfrak{T}:\left(\zeta^k(g_t - h_t^k) - PY_t^i\right)(w) \neq 0 \cap \left(\gamma^k \Gamma_t - PY_t^i + S_t\right)(w) \neq 0\right\}
\end{equation}
has positive measure.

We further define:
\begin{subequations}
\begin{align}
    B_{\text{pos}} &:= \left\{ (w, t) \in \Omega \times \mathfrak{T}:\left(\zeta^k(g_t - h_t^k) - PY_t\right)(w) > 0\right\}, \quad \text{and} 
    \\
    B_{\text{neg}} &:= \left\{ (w, t) \in \Omega \times \mathfrak{T}:\left(\zeta^k(g_t - h_t^k) - PY_t\right)(w) < 0\right\}.
\end{align}
\end{subequations}
At least one of the two sets defined above has positive measure. 
Suppose that $B_{\text{pos}}$ has positive measure. Define the processes 
\begin{align}
    \omega_t^g := \left(\zeta^k(g_t^i - h_t^k) - PY_t^i\right)_+\,, \quad \text{and}\quad
    \omega_t^\Gamma := \left(\gamma^k \Gamma_t^i - PY_t^i + S_t^{\bm{\mu}}\right)\,.
\end{align}
As all processes on the RHS of both expressions are  $\Gg^i$-adapted,  $\omega_t^g, \omega_t^\Gamma$ are as well. This together with the boundedness of $h_t^k, Y_t^i$, and $S_t^{\bm{\mu}}$, and non-negativity of $\omega_t^g$, imply that $(\omega_t^g, \omega_t^\Gamma)\in\mcA^i$.

Therefore, from \eqref{eq:GD_gen}, we obtain
\begin{align*}
    \langle \mcD \overline{J^i}, \omega^g \rangle &= \EE\left[ \int_0^T \omega_t^g \left( \zeta^k (g_t^i - h_t^k) - P Y_t^i\right)dt\right]
    \\
    &=  \EE\left[ \int_0^T \left(\zeta^k(g_t^i - h_t^k) - PY_t^i\right)^2 \II\{\zeta^k (g_t^i - h_t^k) - P Y_t^i > 0\} dt\right] > 0. 
\end{align*}

If, on the other hand, $B_\text{pos}$ does not have positive measure, then $B_\text{neg}$ must, as $B$ has positive measure. In this case, set $\omega_t^g = \left(\zeta^k (g_t^i - h_t^k) - P Y_t^i\right)_-$.
Once again, the pair $(\omega^g,\omega^\Gamma)\in\mcA^i$. In a manner similar to the above, we obtain that $\langle \mcD \overline{J^i}, \omega^g \rangle < 0$.  Hence, we have found an admissible $\omega^g$ such that $\langle \mcD \overline{J^i}, \omega^g \rangle \ne 0$.

Next, from \eqref{eq:GD_trade}, we obtain
\begin{align*}
    \langle \mcD \overline{J^i}, \omega^\Gamma \rangle &= \EE\left[ \int_0^T \omega_t^\Gamma \left( \gamma^k \Gamma_t^i + S_t^{\bm{\mu}} - P Y_t^i\right)dt\right]  
    =  \EE\left[ \int_0^T \left( \gamma^k \Gamma_t^i + S_t^{\bm{\mu}} - P Y_t^i\right)^2 dt\right]  > 0
\end{align*}
Thus $\langle \mcD \overline{J^i}(g^i, \Gamma^i), \omega \rangle \neq 0$, and the proof of necessity is complete.

\subsection{Proof of \Cref{prop:existence}}\label{Proof:prop:existence}
To establish existence, we introduce an `auxiliary' problem that falls under the framework in \cite{carmona2013probabilistic} (our current problem does not due to the endogenous formulation of the price through the clearing condition \eqref{eq:infinite_player_clearing_condition}). This auxiliary problem results in an FBSDE equivalent to \eqref{eqn:FBSDE-full}, and therefore, establishing the existence of a fixed point to this auxiliary problem suffices to establish the existence of a consistent mean field distribution for our original problem.


Consider, from the perspective of a regulated firm $i$ in sub-population $k$, the optimal control problem 
\begin{subequations}
\label{eqn:equiv-MFG}
\begin{equation}
    \inf_{\tilde{g}^i, \tilde{\Gamma}^i \in \mcA^i} \EE\left[\int_0^T \left(\tfrac{\zeta^k}{2}(\tilde{g}_t^i - h_t^k)^2 dt + \tfrac{\gamma^k}{2}(\tilde{\Gamma}^{i}_t)^2 + \tilde{\Gamma}_t^i\,S(t, \bm{\tilde{\mu}}) \right)dt 
    + P\, F^\prime_\delta(R^k - \tilde{X}_T^i)\right],
\end{equation}
subject to 
\begin{align}
    d\tilde{X}_t^i &= (\tilde{g}_t^i + \tilde{\Gamma}_t^i)dt + \sigma\,dW_t^i, \\
    \color{black}
    \tilde{S}(t, \tilde{\bm{\mu}}) &= \color{black} P \sum_{k\in\mcK} \eta_k \int Y^{(k)}(t, x; \tilde{\bm{\mu}})\; \tilde{\mu}_t^{(k)}(dx)\,,
    \label{eqn:S-in-equiv-MFG}
\end{align}
\end{subequations}


where 
\[
\eta_k=
\frac{\frac{\pi_{k}}{\gamma^{k}}}
{\sum_{k' \in\mcK} \frac{\pi_{k'}}{\gamma^{k'}}}\in(0,1) ,\qquad \forall\,k\in\mcK,
\]

The above formulation is equivalent to the original problem (given by \eqref{eq:infinite_player_clearing_condition}-\eqref{eq:problem_statementMF}) with a pre-specified SREC price that coincides with the market clearing SREC price \eqref{eq:equilibrium_price} where firms trade as per \eqref{eq:optGamma}.

As a result, \eqref{eqn:equiv-MFG} falls within the framework in \cite{carmona2013probabilistic}, with the minor change of allowing for sub-populations of agents.


\textcolor{black}{
It should be noted that this auxiliary problem is equivalent to the original problem  but with $S_t^{\bm{\mu}}$ specified as in \eqref{eqn:S-in-equiv-MFG}. This slight reformulation allows us to more directly apply previous work in this field.}

From here, we may follow the same variational analysis approach as in \Cref{sec:variationalAnalysis} to solve this MFG, which results in the optimal controls \eqref{eq:optG} and \eqref{eq:optGamma} with \eqref{eqn:S-in-equiv-MFG} substituted for the price. Alternatively, we could follow the probabilistic method of solving MFGs (outlined in \cite{carmona2013probabilistic}). In this case, the probabilistic method involves expressing the Hamiltonian, optimizing it with respect to the controls, and substituting those controls back into the forward state equation. We then define a backwards SDE through the derivative of the terminal reward and the Hamiltonian (with respect to the forward state). Ultimately, with a rescaling of the backwards SDE and the co-adjoint process, as well as the substitution of \eqref{eqn:S-in-equiv-MFG}, this results in the FBSDE \eqref{eqn:FBSDE-full}.


The solution to \eqref{eqn:equiv-MFG} is given by a mean field distribution $\bm{\tilde{\mu}}$ and a progressively measurable triple $(\tilde{X}^i, \tilde{Y}^i, \tilde{Z}^i) = (\tilde{X}^i_t, \tilde{Y}^i_t, \tilde{Z}^i_t)_{t \in \mfT}$ that satisfies \eqref{eqn:FBSDE-full} 
such that $\tilde{\mu}_t^{(k)}$ coincides with $\mathcal{L}(\tilde{X}^i_t)$ for all $i \in \mfN_k$, $k \in \mcK$.



By the same argument as in \Cref{rmk:markov_lipschitz_adjoint}, $\tilde{Y}_t$ is Markov and Lipschitz
in $(t, \tilde{X_t})$. Therefore, from Theorem 3.2 in \cite{carmona2013probabilistic}, we obtain the existence of a solution to the auxiliary problem. The Markov and Lipschitz properties of the adjoint process $\tilde{Y_t}$ is required to satisfy the assumptions of Theorem 3.2  in \cite{carmona2013probabilistic} (in particular, A.5).  The optimal controls take their form as in \eqref{eq:optG} and \eqref{eq:optGamma}, with $\tilde{Y}_t^i$ replacing $Y_t^i$ as appropriate. The clearing condition still holds, by design, through our choice of $\tilde{S}_t$.

This immediately implies the existence of a solution to our original problem, as the MV-FBSDEs characterizing its solution is also \eqref{eqn:FBSDE-full}.
Therefore, we have established the existence of a mean-field distribution $\bm{\mu}$ and a progressively measurable triple $(X^i, Y^i, Z^i) = (X_t^i, Y_t^i, Z_t^i)_{t \in \mathfrak{T}}$ that satisfies \eqref{eqn:FBSDE-full}, such that $\mu_t^{(k)}$ coincides with $\mathcal{L}(X_t^i)$ for all $i \in \mfN_k$, for all $k \in \mcK$. In particular, this solution is precisely the mean field distribution $\bm{\tilde{\mu}}$ and a progressively measurable triple $(\tilde{X}^i, \tilde{Y}^i, \tilde{Z}^i) = (\tilde{X}^i_t, \tilde{Y}^i_t, \tilde{Z}^i_t)_{t \in \mfT}$ guaranteed by Theorem 3.2 of \cite{carmona2013probabilistic}.

\subsection{Proof of \Cref{prop:uniqueness}}\label{Proof:prop:uniqueness}
Once again, we consider the auxiliary problem \eqref{eqn:equiv-MFG}. Due to the equivalence between the MV-FBSDE associated with this problem and our original problem \eqref{eq:infinite_player_clearing_condition}-\eqref{eq:problem_statementMF} (specifically, that they are both of the form \eqref{eqn:FBSDE-full}), it is clear that demonstrating uniqueness for the auxiliary problem suffices for our purposes. For notational convenience, we omit tildes from the processes in the remainder of the proof, even though these processes all refer to the auxiliary problem.

\textcolor{black}{Before we begin, it is worth nothing that Proposition 3.7 in \cite{carmona2013probabilistic}, which establishes uniqueness for the class of problems considered in that work (which the auxiliary problem falls into) requires a further assumption that the running costs of the agents can be separated into the sum of a function of the control and a function of the mean field distribution. Our problem does not satisfy the additional assumption needed to apply this uniqueness theorem, hence we must develop our own proof, albeit one that is inspired by Proposition 3.7 in \cite{carmona2013probabilistic}}

We prove this via contradiction. Suppose two different measure flows $\bm{\mu}$ and $\hat{\bm{\mu}}$ (along with corresponding progressively measurable triples) both satisfy the auxiliary problem (specifically, this means they both satisfy the fixed point problem associated with \eqref{eqn:FBSDE-full}).   

Consider a generic firm $i$ belonging to sub-population $k$.  Define $\alpha_t^{i, \star} := (g^{i, \star}_t, \Gamma^{i, \star}_t)$
where $g^{i, \star}_t, \Gamma^{i, \star}_t$ are given in \eqref{eq:optG} and \eqref{eq:optGamma}, corresponding to the optimal generation and trading for an agent in sub-population $k$
given the mean field distribution $\bm{\mu}$. Analogously, define $\hat{\alpha}_t^{i, \star} := (\hat{g}^{i, \star}_t, \hat{\Gamma}^{i, \star}_t)$ as the optimal behaviours for the same agent given the mean field distribution $\bm{\hat{\mu}}$. The optimal controls exist for any choice of mean field distribution due to \Cref{prop:optimal_controls}. Their associated state trajectories are denoted $X_t^i$ and $\hat{X}_t^i$ respectively.

Proposition 2.5 in \cite{carmona2013probabilistic}
implies that there exists constants $\lambda^k>0$ s.t.,
\begin{align*}
    &\bar{J}^i(\alpha_t^{i, \star}; \bm{\mu}) + \lambda^k\, \EE\left[\int_0^T \norm{\alpha_t^{i, \star} - \hat{\alpha}_t^{i, \star}}^2 dt \right] 
    \\
    & \qquad \leq \bar{J}^i([\hat{\alpha}_t^{i, \star}, \bm{\hat{\mu}}]; \bm{\mu}) 
    = \EE\left[\int_0^T \left(\tfrac{\zeta^k}{2}(\hat{g}_t^{i, \star} - h_t^k)^2 + \tfrac{\gamma^k}{2} (\hat{\Gamma}_t^{i, \star})^2 + S_t^{\bm{\mu}}\hat{\Gamma}_t^{i, \star} \right) dt + P F^\prime_\delta(R^k - \hat{X}_T) \right]
\end{align*}%
where, borrowing from \cite{carmona2013probabilistic}, the notation $[\hat{\alpha}_t^{i, \star}, \bm{\hat{\mu}}]$ in the cost functional $\bar{J}^i([\hat{\alpha}_t^{i, \star}, \bm{\hat{\mu}}]; \bm{\mu})$ indicates that the measure flow in the drift of $\hat{X}^{i, \star}$ is $\bm{\hat{\mu}}$, whereas the measure flow in the cost functional is $\bm{\mu}$. Subtracting $\bar{J}^i(\hat{\alpha}_t^{i, \star}; \bm{\hat{\mu}}$) from both sides implies
\begin{align*}
   \bar{J}^i(&\alpha_t^{i, \star}; \bm{\mu}) - \bar{J}^i(\hat{\alpha}_t^{i, \star}; \bm{\hat{\mu}}) + 
   \lambda^k\, \EE\left[\int_0^T \norm{\alpha_t^{i, \star} - \hat{\alpha}_t^{i, \star}}^2 dt \right]   
   \\
   &\leq 
   \EE\left[\int_0^T \left(\tfrac{\zeta^k}{2}(\hat{g}_t^{i, \star} - h_t^k)^2 + \tfrac{\gamma^k}{2} (\hat{\Gamma}_t^{i, \star})^2 + S_t^{\bm{\mu}} \hat{\Gamma}_t^{i, \star} \right)dt + P F^\prime_\delta(R^k - \hat{X}_T^i) \right] 
   \\
   &\quad - \EE\left[\int_0^T \left(\tfrac{\zeta^k}{2}(\hat{g}_t^{i, \star} - h_t^k)^2 + \tfrac{\gamma^k}{2} (\hat{\Gamma}_t^{i, \star})^2 + S_t^{\bm{\hat{\mu}}}\hat{\Gamma}_t^{i, \star}  \right)dt + P F^\prime_\delta(R^k - \hat{X}_T^i) \right] 
   \\
   &= \EE\left[\int_0^T \hat{\Gamma}_t^{i, \star} (S_t^{\bm{\mu}} - S_t^{\bm{\hat{\mu}}}) dt \right] \,.
\end{align*}
By exchanging the roles of $\bm{\mu}$ and $\bm{\hat{\mu}}$ and  summing the resulting inequality with the one above, we obtain
\begin{align}
    2 \,\lambda^k \,\EE\left[\int_0^T \norm{\alpha_t^{i,\star} - \hat{\alpha}_t^{i,\star} }^2 dt \right] 
    &\leq  \EE\left[\int_0^T  (S_t^{\bm{\mu}} - S_t^{\bm{\hat{\mu}}}) (\hat{\Gamma}_t^{i, \star} - \Gamma_t^{i, \star}) \, dt \right] 
\end{align}
as the above holds for a generic agent $i$ within sub-population $k$, we may write
\begin{align}
    2 \,\lambda^k \,\EE\left[\int_0^T \norm{\alpha_t^{i,\star} - \hat{\alpha}_t^{i,\star} }^2 dt \right] 
    &\leq  \int_0^T  (S_t^{\bm{\mu}} - S_t^{\bm{\hat{\mu}}}) \left(\int \hat{\Gamma}_t^{\cdot, \star}\hat\mu_t^{(k)}(dx) - \int\Gamma_t^{\cdot, \star}\,\mu_t^{(k)}(dx) \right) \, dt.
\end{align}
Therefore, we now take a weighted sum across all sub-populations, weighted by $\pi_k$, to obtain
\begin{multline}
    \sum_{k\in\mcK}\pi_k\,2 \lambda^k \EE\left[\int_0^T \norm{\alpha^{\cdot,k,\star} - \hat{\alpha}^{\cdot,k,\star}}^2 dt \right]
    \\
    \leq  
    \int_0^T \sum_{k\in\mcK}\pi_k\,\left(\int_{\R} \hat{\Gamma}_t^{\cdot,k,\star} \hat{\mu}_t^{{(k)}}(dx) - \int_{\R} \Gamma_t^{\cdot,k,\star} \mu_t^{ {(k)}}(dx)\right) (S_t^{\bm{\mu}} - S_t^{\bm{\hat{\mu}}}) \, dt \\
    =\int_0^T \left(\sum_{k\in\mcK}\pi_k\,\int_{\R} \hat{\Gamma}_t^{\cdot,k,\star} \hat{\mu}_t^{{(k)}}(dx) - \sum_{k\in\mcK}\pi_k\,\int_{\R} \Gamma_t^{\cdot,k,\star} \mu_t^{ {(k)}}(dx)\right) (S_t^{\bm{\mu}} - S_t^{\bm{\hat{\mu}}}) \, dt =0. \label{eq:contradiction_uniqueness_alternative}
\end{multline}
The notation $\alpha^{\cdot,k}$ portrays the control of a representative agent in sub-population $k$. The final equality holds due to the definition of the optimal generation and the price. The inequality \eqref{eq:contradiction_uniqueness_alternative} forms a contradiction as it implies that $\alpha_t^{\cdot,k,\star} = \hat{\alpha}_t^{\cdot,k, \star}$ for an arbitrary firm in sub-population $k$, for all $k \in \mcK$, which contradicts that there is more than one solution to the MFG. Therefore, uniqueness follows for the auxiliary problem, and hence, the original problem as well.

\subsection{Proof of \Cref{prop:nash_epsilon_rate}}\label{Proof:prop:nash_epsilon_rate}
We begin with the LHS of \eqref{eqn:first-ep-Nash-bound}:
\begin{align}
 \left\lvert J^{i} (g, \Gamma, \bm{g}^{-i,\star} , \bm{\Gamma}^{ -i, \star}) - \overline{J^i}(g, \Gamma; \bm{\mu})
    \right\rvert 
    &=  \left\lvert \EE \left[ \int_0^T \Gamma_t \,(S_t^{\mu^{[N]},-i}- S_t^{\bm{\mu}} )\,dt \right] \right\rvert \nonumber 
    \\
    &\quad \leq \left(\EE \left[ \int_0^T \left\lvert\Gamma_t \right\rvert^2  dt \right]\right)^{\frac{1}{2}} 
    \left(\EE \left[ \int_0^T \left\lvert S_t^{\bm{\mu}} - S_t^{\mu^{[N]},-i}  \right\rvert^2 dt \right]\right)^{\frac{1}{2}}
    \nonumber
    \\
    &\quad \leq M_1\, \left(\EE \left[ \int_0^T \left\lvert S_t^{\bm{\mu}} - S_t^{\mu^{[N]},-i}  \right\rvert^2 dt \right]\right)^{\frac{1}{2}},
    \label{eq:cost_fun_diff}
\end{align}
for some constant $M_1$. The last inequality follows as $\Gamma$ is admissible, and hence square integrable. 

Next, we focus on bounding the remaining expectation in \eqref{eq:cost_fun_diff}.
By definition of $S_t^{\bm{\mu}}$ in \eqref{eq:equilibrium_price_mkv},  $S_t^{\bm{\mu}}$ is insensitive to agent-$i$ deviating from their optimal control, and depends only on the mean field distribution of states across agents (implicitly assuming that all other agents are acting optimally). This observation coupled with the form of $S_t^{\mu^{[N]},-i}$ in \eqref{eqn:finite-S-without-agent-i} implies
\begin{equation}
\EE \left[ \int_0^T 
    \left\lvert S_t^{\bm{\mu}} - S_t^{\mu^{[N]},-i} \right\rvert^2 dt \right] = \EE \left[ \int_0^T \left\lvert 
    S_t^{\bm{\mu}}
    - B_t - C_t \right\rvert^2 dt \right],
\end{equation}
where,
\begin{align}
    B_t&:=P\frac{\displaystyle\sum_{k\in\mcK} \frac{ |\mfN_k^{-i}|}{N \gamma^k} \sum_{j\in\mfN_k^{-i}} \frac{Y^{(k)}(t, X_t^{j,-i}; \bm{\mu})}{|\mfN_k^{-i}|}}
    {\displaystyle\sum_{k\in\mcK} \frac{|\mfN_k^{-i}|}{N\gamma^k}}\,,
    \quad \text{and} \quad
    &C_t:=\frac{\displaystyle\frac{\Gamma_t}{N}}
     {\displaystyle\sum_{k\in\mcK} \frac{|\mfN_k^{-i}|}{N\gamma^k}}\,.
\end{align}
Hence,
\begin{align}
\EE \left[ \int_0^T 
    \left\lvert S_t^{\bm{\mu}} - S_t^{\mu^{[N]},-i} \right\rvert^2 dt \right] &\quad\leq 2\,\EE\left[\int_0^T (S_t^{\bm{\mu}} - B_t)^2 dt \right] + 2\,\EE\left[\int_0^T C_t^2 dt \right] \nonumber \\
    &\quad= 2\,\EE\left[\int_0^T (S_t^{\bm{\mu}} - B_t)^2 dt \right] + o\left(\tfrac{1}{N^2}\right)\,.
\end{align}

Next, we find the rate at which  the remaining expectation above tends to  zero, and in particular, show it is $o(\delta_N^2) + o\left(\frac{1}{N}\right)$. To proceed, first, define 
\begin{align}
    \theta_k &:= \frac{P}{\gamma^k N}, 
    &\rho_k &:=\frac{P\pi_k}{\gamma^k}, \\ 
    \psi &:= \sum_{k \in \mcK} \tfrac{|\mfN_k^{-i}|}{N \gamma^k},\quad \text{and}
    &\beta &:= \sum_{k \in \mcK} \tfrac{\pi_k}{\gamma^k}\,.
\end{align}
 We then have
\begin{equation}\label{eq:cost_difference}
    \EE\left[\int_0^T (S_t^{\bm{\mu}}- B_t)^2 dt \right] 
     = \EE\left[\int_0^T\left( \sum_{k\in\mcK}\left\{ \frac{\theta_k}{\psi} \sum_{j\in \mfN_k^{-i}} Y^{(k)}(t, X_t^{j,-i}; \bm{\mu}) -  \frac{\rho_k}{\beta} \int Y^{(k)}(t, x; \bm{\mu}) \,\mu_t^{(k)}(dx)\right\}\right)^2 dt \right].
\end{equation}

 Next, proceeding similarly to \cite{huang2006large}, we have
\begin{align}
\EE\left[\int_0^T (S_t^{\bm{\mu}}- B_t)^2 dt \right] 
    &= 
    \EE\biggl[\int_0^T\biggl(\sum_{k\in\mcK} \tfrac{\theta_k}{\psi} \sum_{j\in\mfN_k^{-i}} Y^{(k)}(t, X_t^{j,-i}; \bm{\mu}) - \sum_{k\in\mcK} \tfrac{\theta_k}{\psi} \sum_{j\in\mfN_k^{-i}}  Y^{(k)}(t, X_t^{j}; \bm{\mu}) 
    \nonumber \\
    &\hspace{5em}
    +\sum_{k\in\mcK} \tfrac{\theta_k}{\psi} \sum_{j\in\mfN_k^{-i}}  Y^{(k)}(t, X_t^{j}; \bm{\mu}) - \sum_{k\in\mcK} \tfrac{\rho_k}{\beta} \int Y^{(k)}(t, x; \bm{\mu}) \mu_t^{(k)}(dx)\biggr)^2 dt \biggr] 
    \nonumber \\
     &\quad \leq 
     2\,\EE\left[\int_0^T \left(D_t^{(1)}\right)^2 dt\right] 
    +2\,\EE\left[\int_0^T\left(
    D_t^{(2)}\right)^2 dt \right] \,, \label{eq:epsilon_nash_intermediate}
\end{align}
where,
\begin{subequations}
\begin{align}
    D_t^{(1)}&:=\sum_{k\in\mcK} \tfrac{\theta_k}{\psi} \sum_{j\in\mfN_k^{-i}} \left(Y^{(k)}(t, X_t^{j, -i}; \bm{\mu}) -  Y^{(k)}(t, X_t^{j}; \bm{\mu})\right)\,,
    \quad \text{and}
    \\
    D_t^{(2)}&:=\sum_{k\in\mcK} \left[\tfrac{\theta_k}{\psi} \sum_{j\in \mfN_k^{-i}}  Y^{(k)}(t, X_t^{j}; \bm{\mu}) - \tfrac{\rho_k}{\beta} \int Y^{(k)}(t, x; \bm{\mu}) \mu_t^{(k)}(dx)\right]\,.
\end{align}
\end{subequations}
The process $Y_t^j = Y^{(k)}(t, X_t^{j}; \bm{\mu})$ is the adjoint process evaluated  along the state of  agent-$j$ in the \textit{infinite}-player game, i.e., their SREC inventory is obtained using the (infinite-player limit) deterministic price path of $\eqref{eq:equilibrium_price}$, as agent-$i$'s deviation from optimality has no impact to the mean field.

We next  bound each of the terms in \eqref{eq:epsilon_nash_intermediate} separately.
First, 
\begin{align*}
\EE\left[\int_0^T (D_t^{(1)})^2 dt\right] 
&\leq K\,\EE \left[ \int_0^T \sum_{k \in\mcK} \left( \frac{\theta_k}{\psi}\left( \sum_{j\in\mfN_k^{-i}}\left(Y^{(k)}(t, X_t^{j,-i}; \bm{\mu}) - Y^{(k)}(t, X_t^{j}; \bm{\mu})\right) \right) \right)^2 dt \right]
\nonumber 
\\
&\quad \leq K \sum_{k\in\mcK} \frac{\theta_k^2}{\psi^2}\,|\mfN_k^{-i}|\, \EE \left[ \int_0^T  \,\sum_{j\in\mfN_k^{-i}} \left(Y^{(k)}(t, X_t^{j,-i}; \bm{\mu}) - Y^{(k)}(t, X_t^{j}; \bm{\mu})  \right)^2 dt \right]
\nonumber
\\ 
&\quad \leq K \sum_{k\in\mcK} \frac{\theta_k^2}{\psi^2}|\mfN_k^{-i}| \EE \left[ \int_0^T  \sum_{j\in\mfN_k^{-i}}  (L^{(k)})^2 \left(X_t^{j,-i}- X_t^{j}\right)^2  dt \right]
&
\begin{minipage}{0.12\textwidth}
\text{\scriptsize(Lipschitz cont.}\\
\text{\scriptsize\Cref{rmk:markov_lipschitz_adjoint})}
\end{minipage}
&
\nonumber
\\
&\quad= o(\tfrac{1}{N}).
\nonumber
\end{align*}

The last  equality follows from the definition $\theta_k=P/\gamma^k N$, from $|\mfN_k^{-i}|=N_k-\Id_{i\in\mfN_k}$, and from $X_t^{j, -i}\xrightarrow{\;N \rightarrow \infty\;} X_t^j$, for all $j \neq i$. The rate of this convergence is unimportant, as $\EE\left[\int_0^T (D_t^{(1)})^2 dt\right]$ converges to zero at least as fast as $\frac{1}{N}$. In the sequel, we  show that the rate of convergence for $\EE\left[\int_0^T (D_t^{(2)})^2 dt\right]$ is $\frac{1}{N}$, thus, the rate of convergence we find above is sufficiently fast for our purposes.

Next, we have
\begin{equation}
\EE \left[\int_0^T  \left(D_t^{(2)}\right)^2 dt\right] 
=\EE\left[ \int_0^T \left(G^{(1)}_t + G_t^{(2)} \right)^2 dt\right],
\label{eqn:D2-decomp-to-G1-G2}
\end{equation}
by adding and subtracting the term $ \sum_{k\in\mcK} \frac{\theta_k}{\psi}|\mfN_k^{-i}| \int Y^{(k)}(t, x; \bm{\mu}) \mu_t^{(k)}(dx)$, where,
\begin{subequations}
\begin{align}
G_t^{(1)}&:=
\sum_{k\in\mcK} \tfrac{\theta_k}{\psi} \sum_{j\in\mfN_k^{-i}} \left\{ Y^{(k)}(t, X_t^{j}; \bm{\mu}) -  \int Y^{(k)}(t, x; \bm{\mu}) \mu_t^{(k)}(dx)\right\}, \quad\text{and}
\\
G_t^{(2)}&:=
\sum_{k\in\mcK}  \left( \tfrac{\theta_k}{\psi}\, |\mfN_k^{-i}|
- \tfrac{\rho_k}{\beta} \right)
\int Y^{(k)}(t, x; \bm{\mu}) \mu_t^{(k)}(dx). \label{eq:nash_eps_g1_g2}
%
\end{align}%
\end{subequations}%
To bound $G_t^{(2)}$, we have
\begin{align}
    G_t^{(2)} &= \sum_{k\in\mcK}
    \left( \left(\frac{\theta_k}{\psi}\, |\mfN_k^{-i}|
    - \frac{\theta_k}{\beta} |\mfN_k^{-i}|\right) + \left(\frac{\theta_k}{\beta} |\mfN_k^{-i}| - \frac{\rho_k}{\beta} \right)\right)
    \int Y^{(k)}(t, x; \bm{\mu}) \mu_t^{(k)}(dx)
\\
\begin{split}
&= \sum_{k\in\mcK}
    \left( \frac{\theta_k}{\psi\,\beta}\,|\mfN_{k}^{-i}|\, \sum_{k'\in\mcK}\tfrac{1}{\gamma_{k'}}\left(\pi_{k'}-\tfrac{|\mfN_{k'}^{-i}|}{N}\right)
    \right. 
    \\
&\hspace{4em}    \left.
    + \frac{P/\gamma^k}{\sum_{k''\in\mcK}\frac{\pi_{k''}}{\gamma^{k''}}}\left(\tfrac{|\mfN_k^{-i}|}{N} - \pi_k\right)\right)
    \int Y^{(k)}(t, x; \bm{\mu}) \mu_t^{(k)}(dx).
\end{split}
\end{align}
As  $Y^{(k)}$ is bounded by its very definition,  each integral is finite.  Therefore, 
\[
|G_t^{(2)}|\le C\,\left|\tfrac{|\mfN_k^{-i}|}{N} - \pi_k\right|
= o(\delta_N). 
\]
for some $C>0$. If $o(\delta_N)$ is faster than $o(N^{-1})$, then we could replace this with $o(N^{-1})$. However, for the sake of easier notation, as well as conservatism in our rate, we ignore this.

%
%
 Referring back to \eqref{eqn:D2-decomp-to-G1-G2}, we obtain
\begin{equation}
\label{eqn:D2-bound-with-G1-and-delta}
    \EE \biggl[\int_0^T (D_t^{(2)})^2 dt\biggr] 
    \leq 2\, \EE\biggl[\int_0^T \bigl( G_t^{(1)} \bigr)^2 dt \biggr] + o(\delta_N^2)\,.
\end{equation}
All that remains is to find the rate of convergence for the first term on the rhs. To this end,
\begin{align}
    &\EE\biggl[\int_0^T \bigl( G_t^{(1)} \bigr)^2 dt \biggr] \nonumber \allowdisplaybreaks 
    \\
    &\quad
    =\EE\biggl[\int_0^T \biggl(\sum_{k\in\mcK} \tfrac{\theta_k}{\psi} \sum_{j \in \mfN_k^{-i}}\biggl\{Y^{(k)}(t, X_t^{j}; \bm{\mu}) -  \int Y^{(k)}(t, x; \bm{\mu}) \mu_t^{(k)}(dx)\biggr\} \biggr)^2 dt \biggr] \;\qquad\mbox{\scriptsize{(Defn.)}} \nonumber  \allowdisplaybreaks 
    \\ 
    &\quad \leq \frac{K}{\psi^2}\,
    \EE\biggl[\int_0^T \sum_{k\in\mcK}  \theta_k^2  \biggl(\sum_{j \in \mfN_k^{-i}} \biggl\{Y^{(k)}(t, X_t^{j}; \bm{\mu}) -  \int Y^{(k)}(t, x; \bm{\mu}) \mu_t^{(k)}(dx)\biggr\}\biggr)^2 dt \biggr] \nonumber  \allowdisplaybreaks \;\qquad\mbox{\scriptsize{(Cauchy-Schwarz)}}
    \\ 
    &\quad=\frac{K}{\psi^2}\EE\biggl[\int_0^T \sum_{k\in\mcK}  \theta_k^2  \sum_{j \in \mfN_k^{-i}} \biggl( Y^{(k)}(t, X_t^{j}; \bm{\mu}) -  \int Y^{(k)}(t, x; \bm{\mu}) \mu_t^{(k)}(dx)\biggr)^2 dt \biggr] \nonumber  \allowdisplaybreaks \;
\end{align}
The final equality above follows by (i) $Y^{(k)}(t, X_t^j; \bm{\mu})$ and $Y^{(k)}(t, X_t^l; \bm{\mu})$ are independent for $j \neq l$, as they are the adjoint processes evaluated along the states of different agents (with their own idiosyncratic noise) in the \textit{infinite}-player game; (ii) the expectation is over agent-$i$'s time zero $\sigma$-algebra, which excludes other agents' state, thus, $\EE[Y^{(k)}(t, X_t^j; \bm{\mu})] = \int Y^{(k)}(t,x;\bm{\mu})\mu_t^{(k)}(dx)$ for $j \neq i$. As the summation over $j\in\mfN_k^{-i}$ explicitly excludes $i$, all  cross-terms have zero expectations. Continuing on we have,
\begin{align}  
\EE\biggl[\int_0^T \bigl( G_t^{(1)} \bigr)^2 dt \biggr]
    &=\frac{K}{\psi^{2}}\EE\biggl[\int_0^T \sum_{k\in\mcK}  \theta_k^2  \sum_{j \in \mfN_k^{-i}} \left( \int \bigl[Y^{(k)}(t, X_t^{j}; \bm{\mu}) - Y^{(k)}(t, x; \bm{\mu})\bigr] \mu_t^{(k)}(dx)\right)^2  dt \biggr] \nonumber  \allowdisplaybreaks 
    \\ 
    &\quad\leq\frac{K}{\psi^{2}}\EE\biggl[\int_0^T \sum_{k\in\mcK}  \theta_k^2  \sum_{j \in \mfN_k^{-i}} \left( \int \bigl[Y^{(k)}(t, X_t^{j}; \bm{\mu}) - Y^{(k)}(t, x; \bm{\mu})\bigr]^2\, \mu_t^{(k)}(dx)\right)  dt \biggr] \nonumber  \allowdisplaybreaks &  
    \\ 
    &\quad\leq\frac{K}{\psi^{2}}
    \EE\biggl[\int_0^T \sum_{k\in\mcK}  \theta_k^2  \sum_{j \in \mfN_k^{-i}}  \int \bigl[L^{(k)} (x - X_t^j)\bigr]^2 \mu_t^{(k)}(dx)  dt \biggr] \nonumber  \allowdisplaybreaks &\begin{minipage}{0.12\textwidth}
\text{\scriptsize(Lipschitz cont. }\\
\text{\scriptsize\Cref{rmk:markov_lipschitz_adjoint})}
\end{minipage}
    \\ 
    &\quad=2 \frac{K}{\psi^{2}} \sum_{k\in\mcK}  (\theta_k\,L^{(k)})^2  \sum_{j \in \mfN_k^{-i}}  \text{Var}\left[X_t^j\right] 
    \nonumber
    \\
    &\quad=o\left(\tfrac{1}{N}\right)
    \nonumber
\end{align}

Putting this result together with \eqref{eqn:D2-bound-with-G1-and-delta}, $\EE [\int_0^T (D_t^{(2)})^2 dt]$ converges to zero at the rate $\frac{1}{N} + \delta_N^2$.

Following the chain of expressions back up to \eqref{eq:cost_fun_diff}, shows that the rate specified in the statement of the proposition holds.

\section{Numerical solution algorithm for McKean-Vlasov FBSDE} \label{appendix:algorithm}

In this appendix, we detail the steps of our numerical algorithm to solve \eqref{eqn:FBSDE-full}. We first comment on the general structure of our solution to provide an intuition for the methodology used, and then describe the algorithm in great detail.

Recall that we aim to find a mean field distribution $\bm{\mu}$ and a progressively measurable triple $(X_t^i, Y_t^i, Z_t^i)$ that solves \eqref{eqn:FBSDE-full} such that the optimal controls given the choice of $\bm{\mu}$ imply a controlled state process $X_t^i$ over $i$ with the distribution $\bm{\mu}$. Mathematically, this means our solution should have $\mu_t^{(k)}$ coinciding with $\mathcal{L}(X_t^i)$ for all $i \in \mcK$, for all $k \in \mcK$. The $\bm{\mu}$ that does so is a fixed point. At the fixed point, we replace $\mu_t^{(k)}$ with $\mathcal{L}(X_t^i)$, and turn \eqref{eqn:FBSDE-full} into a McKean-Vlasov FBSDE.

Our implementation to solve this fixed point problem is a simple iterative scheme. We begin with an initial guess for $\bm{\mu}$ (denoted by $\bm{\mu}^{(0)}$) and solve \eqref{eqn:FBSDE-full}. From this, we fully characterize the firm's optimal controls through \eqref{eq:optG} and \eqref{eq:optGamma} as well as the price process $S_t^{\bm{\mu}}$. These quantities imply a distribution of $X_t^i$ when all firms are behaving according to their mean field strategy, which we denote as $\bm{\mu}^{(1)}$. If $\bm{\mu}^{(0)}$ and $\bm{\mu}^{(1)}$ are sufficiently close, then we have obtained the solution we seek. If not, we repeat the process, solving the FBSDE given $\bm{\mu}^{(1)}$. This process continues until we have convergence between the mean field distributions, at which point we output the results.

This process is summarized in \Cref{fig:scheme_diagram}.

\tikzstyle{block1} = [rectangle, draw, thick,fill=blue!10, text width = 2.5cm, text centered, rounded corners, minimum width = 2.5cm]
\tikzstyle{block2} = [rectangle, draw, thick,fill=blue!10, text width = 5.5cm, text centered, rounded corners, minimum width = 5.5cm]
\tikzstyle{block3} = [rectangle, draw, thick,fill=blue!10, text width = 6.6cm, text centered, rounded corners, minimum width = 6.6cm]
\tikzstyle{line} = [draw, -latex']
\begin{figure}[h]
    \centering
    \begin{tikzpicture}
        \node [block2] at (-4, 0) (init) {Initialize mean field distributions\\ $\bm{\mu^{(q)}} \coloneqq \{(\mu_t^{(q), k})_{t\in\mfT}\}_{k \in \mcK}$};
        \node [block2] at (4, 0) (solve) {Solve \eqref{eqn:FBSDE-full} given $\bm{\mu^{(q)}}$, for all $k \in \mcK$. \\
        Output $S_t^{\bm{\mu}}, X_t^i, Y_t^i, g_t^{i, \star}, \Gamma_t^{i, \star}$};
        \node [block3] at (4, -4) (law) {Set $\mu_t^{(q+1), k} = \mathcal{L}(X_t^i)$ for all $k \in \mcK$ \\
        Set $\bm{\mu^{(q+1)}} = \{(\mu_t^{(q, k})_{t\in\mfT}\}_{k \in \mcK}$};
        \node [block1] at (-5, -4) (solution) {Output controls, mean field distribution, and price path};
        \draw[line] (init.east)   -- (solve) node[at start, above right] {};
        \draw[line] ([xshift = +1cm]solve.south) -- ([xshift = +1cm]law.north) {};
        \draw[dashed,-latex] (law.west) -- (solution) node[above, xshift = 3.5cm] {\footnotesize if $\bm{\mu^{(q+1)}}$,  $\bm{\mu^{(q)}}$  close};
        \draw[dashed,-latex] ([xshift = -1cm]law.north) -- node[above = 5pt, left=3pt, align=center]{\footnotesize if $\bm{\mu^{(q+1)}}$,  $\bm{\mu^{(q)}}$ not close \\ \footnotesize set $q = q+1$} ([xshift = -1cm]solve.south) {};
    \end{tikzpicture}
    \caption{Numerical Scheme Diagram} \label{fig:scheme_diagram}
\end{figure}
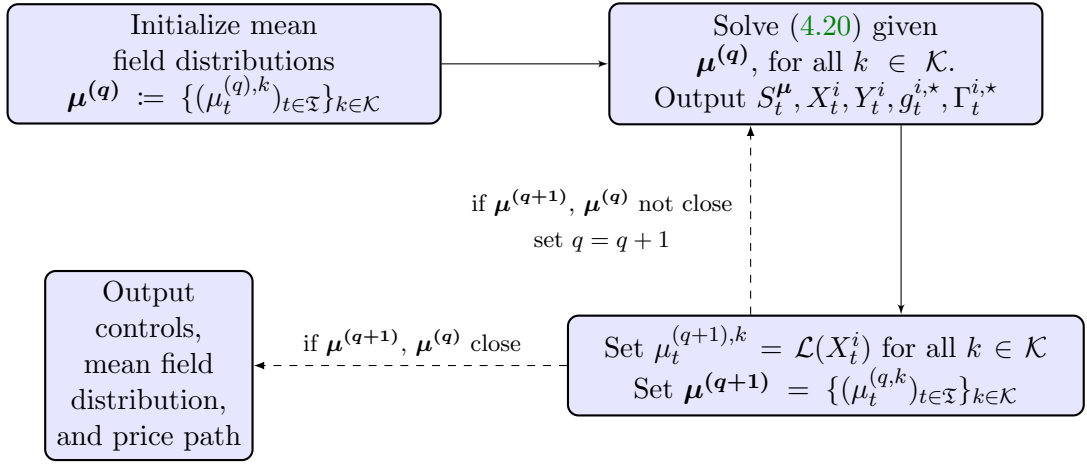

We expand upon this basic explanation below. 
\begin{enumerate}
\item \textbf{Initialization and Setup}

\begin{enumerate}[label=\emph{\Alph*.}]
    \item Set $q = 0$. Define $\epsilon > 0$ as a tolerance limit.
    \item Choose a grid of $X$ (denoted by $\mathfrak{X} :=\{x_1, x_2, ..., x_d\}$), and a grid of time-steps $\mathfrak{T^\prime} := \{0 = t_1, t_2, ..., t_m = T\}$, where $t_j = j \Delta t$.  
    \item For all $k \in \mcK$, choose an initial distribution for the state of agents at $t = 0$ (denoted by $\xi^k$). 
    
    \textbf{Note:}The algorithm proceeds in parallel for all sub-populations, so we omit `for all $k \in \mcK$' henceforth. 
    \item Initialize $\bm{\mu}$. We assume that all agents generate at their baseline $h_t^k$ and do not take part in the SREC market (trading is 0). This assumption, along with the initial distribution of SRECs implied by $\xi^k$, implies a distribution of SRECs across agents at each time-step (for each sub-population) through \eqref{eq:state_dynamics} which is denoted by $\bm{\mu}^{(0)}$.
    \item Initialize $Y_{t_j}$ for all $t_j \in \mfT^\prime$ and all $x \in \mathfrak{X}$ under the assumption that all agents generate at their baseline $h_t^k$ and do not take part in the SREC market. 
\end{enumerate}

\item \textbf{Iteratively update} \label{step:iterate}
    \begin{enumerate}[label=\emph{\Alph*.}]
    \item Let $q = q + 1$
    \item Update $Y_t$:

    For $t_j \in \mfT^\prime$, from $T$ to $0$:

\begin{itemize}
    \item Compute the mean and variance of $X_{t_{j+1}} | X_{t_j} = x$ through \eqref{eq:HetFBSDE_fwd} using the previous update of $Y_t$, $\bm{\mu}^{(q-1)}$.
    \item For all $x_u \in \mathfrak{X}$, assume that $Y_{t_{j+1}} = Y^{(k)}(t_{j+1}, x_u; \bm{\mu}^{(q-1)})$ is locally linear in $x_u$. Use this (in conjunction with \eqref{eq:HetFBSDE_bwd} and \eqref{eq:state_dynamics}) to analytically calculate $Y_{t_j}^k$ for each point $x \in \mathfrak{X}$.
\end{itemize}

\item Update the mean field distribution to find $\bm{\mu}^{(q)}$:

For $t_j \in \mfT^\prime$, from $0$ to $T$:
For $k \in \mcK$
\begin{itemize}
    \item Use the updated $Y_{t_j}$ along with \eqref{eq:HetFBSDE_fwd} to characterize the mean and variance $\mu_{t_j}^{(q), k}$
\end{itemize}

Denote $\bm{\mu}^{(q)} = \{(\mu_{t_j}^{(q), k})_{t\in \mfT^\prime}\}_{k \in \mcK}$

\item Calculate $D \coloneqq d(\bm{\mu}^{(q)}, \bm{\mu}^{(q-1)})$ for some notion of distance between measures, $d$. 

If $D > \epsilon$, return to \Cref{step:iterate}.

Else, proceed to \Cref{step:output}. 
    \end{enumerate}
\item \textbf{Output} \label{step:output}

Use the calculated values to compute \eqref{eq:gen_nec_cond}, \eqref{eq:trade_nec_cond}, \eqref{eq:equilibrium_price}.
\end{enumerate}

\end{document}